\documentclass{article}
\usepackage{amsfonts}
\usepackage{mathrsfs}
\usepackage{amsthm}
\usepackage{amssymb}
\usepackage{amsmath}
\usepackage{enumerate}
\usepackage{braket}
\usepackage{graphicx}
\usepackage{float}

\theoremstyle{plain}
\newtheorem{theorem}{Theorem}[section]
\newtheorem{proposition}[theorem]{Proposition}
\newtheorem{corollary}[theorem]{Corollary}
\newtheorem{lemma}{Lemma}[section]
\theoremstyle{definition}
\newtheorem{definition}{Definition}[section]

\theoremstyle{remark}
\newtheorem{remark}{\textbf{Remark}}[section]

\theoremstyle{definition}

\usepackage{chngcntr}
\counterwithout{table}{section}

\numberwithin{equation}{section}

\title{Harmonic Reduction of Spin-Coupled Grover Walks on Hypercubes via Quantum Bernoulli Noises}
\author{Ce Wang\thanks{email: cewangwhu@163.com; orcid:0000-0002-2916-8478} \\
Shanghai Institute for Mathematics and Interdisciplinary Sciences (SIMIS)\\
Shanghai 200433, China}
\date{}

\begin{document}
\maketitle

\begin{abstract}
We consider spin-coupled Grover walks on hypercubes within the
framework of quantum Bernoulli noises (QBN).  We identify
parity-dependent harmonic subspaces generated by the central Hamming
layers: in even dimension they give a fixed-point subspace, whereas in
odd dimension the walk reduces on the harmonic subspace to an explicit
eight-dimensional reduced dynamics.  Using symmetric-group
representation theory, we characterize the corresponding isotypic
structure and determine the reduced spectrum.  The harmonic reduction
also yields explicit long-time probability behavior and shows that
the resulting states may occupy an asymptotically vanishing fraction
of the hypercube while still spanning its full Hamming diameter.
\end{abstract}

\vskip 2mm

\vskip 2mm


\section{Introduction}
\label{sec:introduction}

Quantum walks are quantum analogues of classical random walks and
provide a general framework for quantum dynamics on graphs, with
applications in quantum computation, quantum algorithms, and
mathematical physics.  Since their early introduction in
\cite{Aharonov1993}, they have been shown to provide a universal model
of quantum computation \cite{Childs2009} and exponential speedups for
certain algorithmic problems \cite{ChildsEtAl2003}.  Among the graph
geometries studied in this context, the hypercube is particularly
important because of its high symmetry and rich combinatorial and
spectral structure, with applications to mixing, hitting, quantum
search, and decoherence
\cite{MooreRussell,Kempe2005,ShenviKempeWhaley2003,AlagicRussell}.

Symmetry plays a central role in the analysis and reduction of quantum
walks on highly regular graphs.  On the hypercube, permutation
symmetry reduces the dynamics to variables determined by Hamming
weight or distance \cite{MooreRussell}, while more general
group-theoretic constructions lead to walks on quotient graphs
\cite{KroviBrun2007}.  Related approaches use association schemes and
spectral representations on Hamming graphs
\cite{GriffithsMano2026}, together with the classical raising--lowering
structure of the Boolean lattice and its relation to the Hamming and
Johnson schemes \cite{Feinsilver2012}.  The invariant structure of a
quantum walk also depends strongly on the choice of coin and on
additional internal degrees of freedom.  For Grover walks, spectral
degeneracies can produce nontrivial stationary subspaces
\cite{InuiKonishiKonno2004,KonnoObataSegawa2013}, while spin-coupled
walks incorporate internal variables that modify the effective walker
dynamics \cite{DanaciEtAl2021}.  Finite-dimensional invariant
reductions can likewise arise in particular symmetric settings; for
example, a coined walk on the complete bipartite graph admits an
eight-dimensional invariant reduction whose spectrum can be analyzed
explicitly \cite{BezerraSantosPortugal2024}.  In contrast with such
model-specific reductions, the present work seeks a systematic
description of the central invariant dynamics of the spin-coupled
hypercube walk in arbitrary dimension.

Quantum Bernoulli noises (QBN) are annihilation and creation operators
acting on Bernoulli functionals and satisfying the equal-time
canonical anti-commutation relations.  Their discrete-time theory was
developed in \cite{WangChaiLu2010}, quantum walks in terms of QBN were
introduced in \cite{WangYe2016}, and the QBN realization of hypercube
walks was developed in \cite{WangHypercube2022}.  The spin-coupled
Grover walk considered here was formulated in this framework in
\cite{WangSpectral2026}.  These developments naturally lead to the
question of whether the QBN operator algebra itself can reveal and
organize the invariant subspaces of the walk, rather than merely
provide an alternative representation of its evolution.

We show that the collective QBN annihilation and creation operators
answer this question through their harmonic kernels on the central
Hamming layers.  Their commutation relation reveals a parity-dependent
harmonic mechanism at the center of the hypercube.  For \(d=2m\), the
central harmonic space generates a three-layer fixed-point subspace
which, for \(m\ge2\), exhausts all fixed states with this support.  For
\(d=2m+1\), the collective creation operator induces an isometric
isomorphism between the two central harmonic spaces, and the
associated QBN lifts generate a four-layer invariant subspace whose
nontrivial dynamics is encoded by an explicit eight-dimensional
reduced matrix.  Thus a single QBN mechanism produces the distinct
even- and odd-dimensional reductions and applies uniformly in
arbitrary dimension.

The odd-dimensional construction also admits a precise
representation-theoretic characterization.  The central harmonic
kernel is isomorphic to the Specht module \(S^{(m+1,m)}\), and the
resulting four-layer harmonic subspace is the full corresponding
isotypic component.  The multiplicity \(8\), determined by the
symmetric-group branching rule, coincides exactly with the dimension
of the reduced dynamical factor.  Together with the
even-dimensional fixed-point completeness theorem, this shows that
the QBN construction identifies a complete parity-dependent reduction
of the central dynamics, rather than a particular low-dimensional
invariant sector.  We further determine the spectrum of the reduced
dynamics and derive consequences for its long-time and spatial
behavior: harmonic position distributions remain asymptotically
maximally separated from the uniform measure, while their supports may
occupy an asymptotically vanishing fraction of the hypercube and still
have maximal Hamming diameter.

The paper is organized as follows.  Section~\ref{sec:pre} recalls the
QBN framework and the spin-coupled Grover walk.
Section~\ref{sec:incidence-harmonic} develops the central harmonic
structure.  Section~\ref{sec:harmonic-subspaces} constructs the
parity-dependent invariant subspaces and derives the reduced dynamics,
representation-theoretic structure, and spectrum.
Section~\ref{sec:long-time-geometry} studies the resulting long-time
and spatial behavior, and the final section contains concluding
remarks.

\section{Preliminaries}
\label{sec:pre}
\subsection{QBN construction and hypercube realization}
\label{subsec:qbn-hypercube}
In this section, we recall the QBN construction and its realization on a
hypercube, and then formulate the spin-coupled Grover walk used in this
paper; see \cite{WangChaiLu2010,WangYe2016,WangHypercube2022} for
further details.

Let
\(
\Omega=\{-1,1\}^{\mathbb N}
\)
be the Bernoulli space with the canonical product measure \(\mathbb P\).
The space of square-integrable Bernoulli functionals is denoted by
\(
\mathfrak h=L^2(\Omega,\mathcal F,\mathbb P).
\)
Let \(\Gamma\) be the finite power set of \(\mathbb N\). It is known that
\(\mathfrak h\) has an orthonormal basis (ONB)
\(\{Z_\sigma\mid \sigma\in\Gamma\}\) of the form
\(Z_\emptyset=1\) and
\begin{equation*}
%\label{h_ONB}
Z_\sigma
=
\prod_{j\in\sigma} Z_j,
\qquad
\sigma\in\Gamma,\quad \sigma\neq\emptyset,
\end{equation*}
where \(Z=(Z_n)_{n\ge0}\) is a sequence of Bernoulli functionals.

For each \(k\in\mathbb N\), the annihilation operator \(\partial_k\)
and the creation operator \(\partial_k^*\) act on the ONB
\(\{Z_\sigma:\sigma\subset\mathbb N,\ |\sigma|<\infty\}\) as
\begin{equation}
\label{eq:QBN}
\partial_k Z_\sigma
=
\mathbf 1_\sigma(k)Z_{\sigma\setminus\{k\}},
\qquad
\partial_k^* Z_\sigma
=
\bigl(1-\mathbf 1_\sigma(k)\bigr)Z_{\sigma\cup\{k\}}.
\end{equation}
For each \(k\), these operators satisfy the equal-time canonical
anti-commutation relations (CAR)
\begin{equation*}
%\label{eq:CAR}
\partial_k^2=(\partial_k^*)^2=0,
\qquad
\partial_k\partial_k^*
+
\partial_k^*\partial_k
=
I.
\end{equation*}
For distinct coordinates they satisfy
\[
[\partial_j,\partial_k]
=
[\partial_j^*,\partial_k^*]
=
[\partial_j,\partial_k^*]
=
0,
\qquad
j\neq k.
\]
The operator family
\(
\{\partial_k,\partial_k^*:k\in\mathbb N\}
\)
is referred to as quantum Bernoulli noises (QBN).

We denote by \(\mathfrak h_n\) the subspace of \(\mathfrak h\) spanned by
\(\{Z_\sigma:\sigma\in\Gamma_n\}\),
where \(\Gamma_n\) is the power set of
\(
\mathbb N_n=\{0,1,\ldots,n\}.
\)
Since \(\Gamma_n\subset\Gamma\) and
\(|\Gamma_n|=2^{n+1}\), the space \(\mathfrak h_n\) is a
\(2^{n+1}\)-dimensional closed subspace of \(\mathfrak h\).
Moreover, for every \(k\in\mathbb N_n\), both
\(\partial_k\) and \(\partial_k^*\) leave \(\mathfrak h_n\) invariant.

We identify the vertex set of the \((n+1)\)-dimensional hypercube
with \(\Gamma_n\): two vertices \(\sigma,\tau\in\Gamma_n\) are adjacent
whenever
\[
|\sigma\triangle\tau|=1.
\]
By \eqref{eq:QBN}, for every \(k\in\mathbb N_n\),
\[
(\partial_k^*+\partial_k)Z_\sigma
=
Z_{\sigma\triangle\{k\}}.
\]
Thus
\begin{equation}
\label{shift_operator}
\Xi_k
:=
\partial_k+\partial_k^*,
\qquad
k\in\mathbb N_n,
\end{equation}
acts as the coordinate-\(k\) shift. Each \(\Xi_k\) is a self-adjoint
unitary, and the family
\(\{\Xi_k:k\in\mathbb N_n\}\) is commuting; see
\cite{WangHypercube2022} for the QBN realization of the hypercube.

Let
\(
\mathcal K
\) be a finite dimensional Hilbert space.
A family
\[
\mathcal C=\{C_k:0\le k\le n\}
\subset\mathcal B(\mathcal K)
\]
is called a coin operator system if
\(
\sum_{k=0}^n C_k
\)
is unitary and
\[
C_j^*C_k=C_jC_k^*=0,
\qquad
j\ne k.
\]
It follows that
\[
\sum_{k=0}^n C_k^*C_k
=
\sum_{k=0}^n C_kC_k^*
=
I_{\mathcal K}.
\]
A coin operator system admits the following equivalent representation.
There exist a unitary operator \(U\) on \(\mathcal K\) and a resolution
of the identity
\[
\{P_k:0\le k\le n\},
\]
such that
\begin{equation}
C_k=P_kU,
\qquad
0\le k\le n.
\label{eq:coin-system}
\end{equation}
Indeed, one may take
\[
U=\sum_{k=0}^n C_k,
\qquad
P_k=C_kC_k^*.
\]
Conversely, every family of the form \(C_k=P_kU\) is a coin operator
system.
The corresponding QBN hypercube walk is
\begin{equation}
W
=
\sum_{k=0}^{n}
\Xi_k\otimes C_k
\label{eq:qbn-hypercube-walk}
\end{equation}
on
\(\mathfrak h_n\otimes\mathcal K\).
It is unitary; see \cite{WangHypercube2022}.

\subsection{Spin-coupled Grover walk and eigenvalue equation}
\label{sec:model}

We next recall the spin-coupled walk. The construction is motivated by
the conserved-spin mechanism for disorder-free localization
\cite{DanaciEtAl2021} and was formulated mathematically in the QBN
framework in \cite{WangSpectral2026}.

Let \(d\ge2\), and let
\(\mathfrak h_{d-1}\) be the position Hilbert space of the
\(d\)-dimensional hypercube.  From now on we take
\(
\mathcal K
=
\mathbb C^d
\)
as the coin space, with orthonormal basis
\(
\{e_0,\ldots,e_{d-1}\}.
\)
Set
\(
M=2^d-1.
\)
Since
\[
|\mathbb N_M|
=
M+1
=
2^d
=
|\Gamma_{d-1}|,
\]
fix a bijection
\[
o:
\Gamma_{d-1}
\longrightarrow
\mathbb N_M.
\]
Let \(\mathfrak h_M\) be the spin Hilbert space with orthonormal basis
\[
\{Z_\gamma:\gamma\subseteq\mathbb N_M\}.
\]
For each \(\sigma\in\Gamma_{d-1}\), define the spin-flip operator by
\[
\Pi_\sigma^{\mathrm{spin}}Z_\gamma
=
Z_{\gamma\triangle\{o(\sigma)\}},
\qquad
\gamma\subseteq\mathbb N_M.
\]
The total Hilbert space is
\[
\mathcal H
=
\mathfrak h_{d-1}
\otimes\mathcal K
\otimes\mathfrak h_M.
\]

The position-dependent spin rotation is
\[
R
=
\sum_{\sigma\in\Gamma_{d-1}}
|Z_\sigma\rangle\langle Z_\sigma|
\otimes I_{\mathcal K}
\otimes
e^{-i\phi_\sigma\Pi_\sigma^{\mathrm{spin}}},
\qquad
\phi_\sigma
=
\frac{|\sigma|\pi}{d}.
\]
The full evolution operator on \(\mathcal H\) is
\[
\widetilde W
=
\left(
\sum_{k=0}^{d-1}
\Xi_k\otimes C_k\otimes I_{\mathfrak h_M}
\right)R.
\]

The spin-flip operators commute and are self-adjoint unitaries.
Hence they admit a common orthonormal eigenbasis indexed by
\[
\mathbf s
=
(s_\sigma)_{\sigma\in\Gamma_{d-1}}
\in
\{\pm1\}^{\Gamma_{d-1}}.
\]
Write the corresponding common eigenvector as
\(\mathcal S(\mathbf s)\), so that
\[
\Pi_\sigma^{\mathrm{spin}}
\mathcal S(\mathbf s)
=
s_\sigma\mathcal S(\mathbf s).
\]
For each \(\mathbf s\), the subspace
\[
\mathcal H_{\mathbf s}
=
\mathfrak h_{d-1}
\otimes\mathcal K
\otimes
\operatorname{span}\{\mathcal S(\mathbf s)\}
\]
is invariant under \(\widetilde W\).

For
\(
\mathbf s\in\{\pm1\}^{\Gamma_{d-1}}
\)
and
\(
\sigma\in\Gamma_{d-1},
\)
define the operator family
\[
\mathfrak C_{\mathbf s}^{(\sigma)}
=
\{C_{\mathbf s,k}^{(\sigma)}:k=0,\ldots,d-1\}
\]
by
\[
C_{\mathbf s,k}^{(\sigma)}
:=
e^{-i|\sigma|\pi s_\sigma/d}\,C_k.
\]
Since the common phase factor has unit modulus,
\(\mathfrak C_{\mathbf s}^{(\sigma)}\) is again a coin operator system.

On the invariant subspace \(\mathcal H_{\mathbf s}\), the full evolution
reduces to the position-dependent quantum walk
\begin{equation}
\label{eq:W_s_reduced}
\widetilde W_{\mathbf s}
=
\sum_{k=0}^{d-1}
\sum_{\sigma\in\Gamma_{d-1}}
|Z_{\sigma\triangle\{k\}}\rangle
\langle Z_\sigma|
\otimes
C_{\mathbf s,k}^{(\sigma)}.
\end{equation}
Hence a vector
\[
\Phi
=
\sum_{\sigma\in\Gamma_{d-1}}
Z_\sigma\otimes u_\sigma
\]
is a fixed point of \(\widetilde W_{\mathbf s}\) if and only if
\begin{equation}
\label{eq:fp-general}
u_\sigma
=
\sum_{k=0}^{d-1}
e^{-i|\sigma\triangle\{k\}|\pi
s_{\sigma\triangle\{k\}}/d}
C_k
u_{\sigma\triangle\{k\}},
\qquad
\sigma\in\Gamma_{d-1}.
\end{equation}

For the Grover coin system, let
\[
C_k=P_kG,
\qquad
k=0,1,\ldots,d-1,
\]
where
\[
G
=
2|u_d\rangle\langle u_d|-I_{\mathcal K},
\qquad
u_d
=
\frac{1}{\sqrt d}
\sum_{k=0}^{d-1}e_k,
\]
and
\[
P_k
=
|e_k\rangle\langle e_k|
\]
is the orthogonal projection onto the \(k\)-th coin direction
\(\operatorname{span}\{e_k\}\).

Let
\[
\Phi
=
\sum_{\sigma\in\Gamma_{d-1}}
Z_\sigma\otimes u_\sigma,
\qquad
u_\sigma
=
\sum_{k=0}^{d-1}u_\sigma^{(k)}e_k.
\]
Then \(\Phi\) is an eigenvector of
\(\widetilde W_{\mathbf s}\) with eigenvalue
\(\lambda\in\mathbb C\) if and only if
\begin{equation}
\lambda u_\sigma
=
\sum_{k=0}^{d-1}
e^{-i|\sigma\triangle\{k\}|\pi
s_{\sigma\triangle\{k\}}/d}
C_k
u_{\sigma\triangle\{k\}},
\qquad
\sigma\in\Gamma_{d-1}.
\label{eq:eig-general}
\end{equation}
Since \(C_k=P_kG\), the \(k\)-th coin component satisfies
\begin{equation}
\lambda u_\sigma^{(k)}
=
e^{-i|\sigma\triangle\{k\}|\pi
s_{\sigma\triangle\{k\}}/d}
\bigl(
G u_{\sigma\triangle\{k\}}
\bigr)^{(k)},
\label{eq:eig-component-general}
\end{equation}
for
\(
k=0,\ldots,d-1,
\sigma\in\Gamma_{d-1}.
\)
In particular, the fixed-point equation is obtained from
\eqref{eq:eig-general} by setting \(\lambda=1\).
%%%%%%%%%%%%%%%%%%%%%%%%%%%%%%%%%%%%%%%%%%%%%%%%%%%%%%%%%%%%%%%%%%%

\section{QBN harmonic structure on the hypercube}
\label{sec:incidence-harmonic}

In this section, we formulate the harmonic structure of the hypercube
in terms of QBN.  We first introduce the number operator and the
collective annihilation and creation operators, and then examine their
behavior on the central Hamming layers.

For \(0\le r\le d\), define the \(r\)-th Hamming layer by
\[
\Omega_r
:=
\{\sigma\in\Gamma_{d-1}:|\sigma|=r\},
\]
and let
\[
V_r
:=
\operatorname{span}\{Z_\sigma:\sigma\in\Omega_r\}
\subset\mathfrak h_{d-1}.
\]

Define the number operator by
\[
\mathcal N
:=
\sum_{k=0}^{d-1}\partial_k^*\partial_k.
\]
By \eqref{eq:QBN},
\[
\mathcal N Z_\sigma
=
|\sigma|Z_\sigma.
\]
Hence the Hamming-layer decomposition is precisely the spectral
decomposition of \(\mathcal N\):
\[
\mathfrak h_{d-1}
=
\bigoplus_{r=0}^{d}V_r,
\qquad
V_r
=
\ker(\mathcal N-rI).
\]
Define the collective annihilation and creation operators by
\[
A^-
:=
\sum_{k=0}^{d-1}\partial_k,
\qquad
A^+
:=
\sum_{k=0}^{d-1}\partial_k^*.
\]
They connect adjacent Hamming layers:
\[
A^-V_r\subseteq V_{r-1},
\qquad
A^+V_r\subseteq V_{r+1},
\]
with the convention
\(
V_{-1}=V_{d+1}=\{0\}.
\)
Moreover,
\[
(A^-)^*=A^+.
\]

The basic algebraic relation between the collective operators and the
number operator is the following.

\begin{proposition}
\label{thm:sl2-qbn}
On \(\mathfrak h_{d-1}\),
\begin{equation}
A^-A^+-A^+A^-
=
dI-2\mathcal N.
\label{eq:qbn-commutator-global}
\end{equation}
In particular, on \(V_r\),
\begin{equation}
\left.
\bigl(A^-A^+-A^+A^-\bigr)
\right|_{V_r}
=
(d-2r)I_{V_r}.
\label{eq:qbn-commutator-layer}
\end{equation}
\end{proposition}

\begin{proof}
Expanding the left-hand side of
\eqref{eq:qbn-commutator-global} gives
\[
\sum_{j,k=0}^{d-1}
\left(
\partial_j\partial_k^*
-
\partial_j^*\partial_k
\right).
\]
For \(j\neq k\), the QBN operators associated with distinct coordinates
commute, so the corresponding terms cancel pairwise.  Hence only the
diagonal terms remain:
\[
A^-A^+-A^+A^-
=
\sum_{k=0}^{d-1}
\left(
\partial_k\partial_k^*
-
\partial_k^*\partial_k
\right).
\]
By the equal-time CAR,
\[
\partial_k\partial_k^*
-
\partial_k^*\partial_k
=
I-2\partial_k^*\partial_k.
\]
Summing over \(k\) yields
\[
A^-A^+-A^+A^-
=
dI
-
2\sum_{k=0}^{d-1}\partial_k^*\partial_k
=
dI-2\mathcal N.
\]
Since \(\mathcal N=rI\) on \(V_r\),
\eqref{eq:qbn-commutator-layer} follows.
\end{proof}

\begin{remark}
Equation~\eqref{eq:qbn-commutator-global} already contains the
even--odd distinction at the center of the hypercube.  If \(d=2m\),
the commutator vanishes on the unique central number subspace \(V_m\).
If \(d=2m+1\), it equals \(I\) on \(V_m\) and \(-I\) on
\(V_{m+1}\).  The central QBN harmonic structures below arise directly
from these two regimes.
\end{remark}

\begin{corollary}
\label{cor:norm-qbn}
For every \(f\in\mathfrak h_{d-1}\),
\begin{equation}
\|A^+f\|^2-\|A^-f\|^2
=
d\|f\|^2
-
2\langle f,\mathcal N f\rangle.
\label{eq:norm-qbn-global-expanded}
\end{equation}
In particular, if \(f\in V_r\), then
\begin{equation}
\|A^+f\|^2-\|A^-f\|^2
=
(d-2r)\|f\|^2.
\label{eq:norm-qbn}
\end{equation}
\end{corollary}

%\begin{proof}
%Since \((A^-)^*=A^+\),
%\[
%\langle f,A^-A^+f\rangle
%=
%\|A^+f\|^2,
%\qquad
%\langle f,A^+A^-f\rangle
%=
%\|A^-f\|^2.
%\]
%Taking the inner product of
%\eqref{eq:qbn-commutator-global} with \(f\) gives
%\[
%\|A^+f\|^2-\|A^-f\|^2
%=
%\left\langle
%f,(dI-2\mathcal N)f
%\right\rangle,
%\]
%which proves \eqref{eq:norm-qbn-global-expanded}.  If \(f\in V_r\), then
%\(\mathcal Nf=rf\), and \eqref{eq:norm-qbn} follows.
%\end{proof}

\begin{corollary}
\label{cor:qbn-ladder-estimates}
Let \(0\le r\le d\).

\begin{enumerate}[(i)]
\item
If \(r<d/2\), then for every \(f\in V_r\),
\begin{equation}
\|A^+f\|^2
\ge
(d-2r)\|f\|^2.
\label{eq:lower-ladder-estimate}
\end{equation}
In particular,
\[
A^+\big|_{V_r}:V_r\longrightarrow V_{r+1}
\]
is injective, and 
\[
A^-\big|_{V_{r+1}}:V_{r+1}\longrightarrow V_r
\]
is surjective.

\item
If \(r>d/2\), then for every \(f\in V_r\),
\begin{equation}
\|A^-f\|^2
\ge
(2r-d)\|f\|^2.
\label{eq:upper-ladder-estimate}
\end{equation}
In particular,
\[
A^-\big|_{V_r}:V_r\longrightarrow V_{r-1}
\]
is injective, and
\[
A^+\big|_{V_{r-1}}:V_{r-1}\longrightarrow V_r
\]
is surjective.
\end{enumerate}
\end{corollary}

\begin{proof}
If \(r<d/2\), then \eqref{eq:norm-qbn} gives
\[
\|A^+f\|^2
=
\|A^-f\|^2
+
(d-2r)\|f\|^2,
\]
and hence
\[
\|A^+f\|^2
\ge
(d-2r)\|f\|^2.
\]
Since \(d-2r>0\), the restriction of \(A^+\) to \(V_r\) is injective.
Because
\[
\left(A^+\big|_{V_r}\right)^*
=
A^-\big|_{V_{r+1}},
\]
and the spaces are finite-dimensional, the latter map is surjective.

If \(r>d/2\), then
\[
\|A^-f\|^2
=
\|A^+f\|^2
+
(2r-d)\|f\|^2,
\]
which gives \eqref{eq:upper-ladder-estimate}.  The injectivity and
surjectivity statements follow in the same way.
\end{proof}

We now define harmonicity in terms of the collective QBN
operators.

\begin{definition}
\label{def:qbn-harmonic}
For \(0\le r\le d\), define the lower and upper harmonic subspaces by
\begin{equation}
\mathscr H_r^-
:=
V_r\cap\ker A^-,
\qquad
\mathscr H_r^+
:=
V_r\cap\ker A^+.
\label{eq:central-harmonic-spaces}
\end{equation}
Vectors in \(\mathscr H_r^-\) and \(\mathscr H_r^+\) are called lower
and upper harmonic vectors, respectively.
\end{definition}

Under the standard identification of \(V_r\) with functions on the
Boolean \(r\)-slice, the lower QBN harmonic condition
\[
A^-f=0
\]
is the usual harmonicity condition on the slice; see
\cite{Filmus2016}.  In what follows, however, the QBN annihilation and
creation conditions in \eqref{eq:central-harmonic-spaces} will be used
directly.

For later use, set
\begin{equation}
q_-
:=
\left\lfloor\frac d2\right\rfloor,
\qquad
q_+
:=
\left\lceil\frac d2\right\rceil.
\label{eq:central-indices}
\end{equation}

The following theorem describes the harmonic structure.

\begin{theorem}
\label{thm:harmonic-qbn}
Let \(d\ge2\).

\medskip
\noindent
{\rm (i)}
If \(d=2m\), then
\begin{equation}
\mathscr H_m^-
=
\mathscr H_m^+.
\label{eq:even-central-harmonic}
\end{equation}
\medskip
\noindent
{\rm (ii)}
If \(d=2m+1\), then
\begin{equation}
A^+:
\mathscr H_m^-
\longrightarrow
\mathscr H_{m+1}^+
\label{eq:odd-harmonic-map}
\end{equation}
is an isometric isomorphism whose inverse is the restriction of \(A^-\).
In particular, if
\[
f\in\mathscr H_m^-,
\qquad
g:=A^+f,
\]
then
\begin{equation}
A^-f=0,
\qquad
A^+f=g,
\qquad
A^-g=f,
\qquad
A^+g=0,
\qquad
\|g\|=\|f\|.
\label{eq:odd-harmonic-transfer}
\end{equation}
\end{theorem}

\begin{proof}
Let first \(d=2m\).  For \(f\in V_m\),
\eqref{eq:norm-qbn} gives
\[
\|A^+f\|^2-\|A^-f\|^2=0.
\]
Hence
\(
\|A^+f\|=\|A^-f\|,
\)
and therefore
\[
A^-f=0
\quad\Longleftrightarrow\quad
A^+f=0.
\]
This proves \eqref{eq:even-central-harmonic}.

Now let \(d=2m+1\) and
\(f\in\mathscr H_m^-\).  Since
\(
A^-f=0,
\)
the commutator identity
\eqref{eq:qbn-commutator-global} restricted to \(V_m\) gives
\[
A^-A^+f=f.
%\label{eq:odd-reconstruction}
\]
Moreover, \eqref{eq:norm-qbn} gives
\[
\|A^+f\|=\|f\|.
%\label{eq:odd-isometry}
\]

Set
\(
g:=A^+f\in V_{m+1}.
\) Then
\(
A^-g=f.
\)
Applying \eqref{eq:norm-qbn} to \(g\in V_{m+1}\), we obtain
\[
\|A^+g\|^2-\|A^-g\|^2
=
-\|g\|^2.
\]
Using
\(
\|A^-g\|=\|f\|=\|g\|,
\)
we conclude that
\(
A^+g=0.
\)
Thus
\(
g\in\mathscr H_{m+1}^+,
\)
and \(A^+\) defines an isometry
\[
\mathscr H_m^-
\longrightarrow
\mathscr H_{m+1}^+.
\]
Conversely, let
\(
g\in\mathscr H_{m+1}^+.
\)
Since \(A^+g=0\), the commutator identity on \(V_{m+1}\) gives
\[
A^+A^-g=g.
\]
Set
\(
f:=A^-g.
\)
Then \(A^+f=g\).  Applying \eqref{eq:norm-qbn} to \(g\) gives
\(
\|f\|=\|g\|.
\)
Applying it to \(f\in V_m\), together with
\(\|A^+f\|=\|g\|=\|f\|\), yields
\[
A^-f=0.
\]
Hence
\(
f\in\mathscr H_m^-.
\)
Therefore the restrictions of \(A^+\) and \(A^-\) are inverse
isometries.
\end{proof}

The theorem exhibits the parity-dependent structure of the central
harmonic subspaces.  The dimensions of the central harmonic spaces can also be obtained.

\begin{corollary}
\label{thm:dim-harmonic-qbn}
Let \(d\ge2\), and let
\[
C_n
:=
\frac{1}{n+1}\binom{2n}{n}
\label{eq:catalan-number}
\]
denote the \(n\)-th Catalan number.  Then
\begin{equation}
\dim\mathscr H_{q_-}^-
=
\dim\mathscr H_{q_+}^+
=
C_{q_+}.
\label{eq:dim-harmonic}
\end{equation}
\end{corollary}

\begin{proof}
By Proposition~\ref{cor:qbn-ladder-estimates},
\[
A^-:
V_{q_-}
\longrightarrow
V_{q_--1}
\]
is surjective.  Therefore, by rank--nullity,
\[
\dim\mathscr H_{q_-}^-
=
\dim V_{q_-}-\dim V_{q_--1}
=
\binom{d}{q_-}
-
\binom{d}{q_--1}.
\label{eq:central-kernel-difference}
\]
Thus
\[
\dim\mathscr H_{q_-}^-
=
C_{q_+}.
\]
The equality
\[
\dim\mathscr H_{q_+}^+
=
\dim\mathscr H_{q_-}^-
\]
follows from Theorem~\ref{thm:harmonic-qbn}.
\end{proof}

%%%%%%%%%%%%%%%%%%%%%%%%%%%%%%%%%%%%%%%%%%%%%%%%%%%%%%%%%%%%%%%%%%%
\section{Parity-dependent harmonic dynamics}
\label{sec:harmonic-subspaces}
%%%%%%%%%%%%%%%%%%%%%%%%%%%%%%%%%%%%%%%%%%%%%%%%%%%%%%%%%%%%%%%%%%%

In this section, we construct invariant position--coin states from the
harmonic vectors obtained above.  The construction is based on four
canonical QBN lifts:
\[
\mathsf L_{\downarrow},
\mathsf L_{\uparrow},
\mathsf L^-,
\mathsf L^+
:
\mathfrak h_{d-1}
\longrightarrow
\mathfrak h_{d-1}\otimes\mathcal K.
\]
They are defined by
\begin{equation}
\begin{aligned}
\mathsf L_{\downarrow}h
&:=
\sum_{k=0}^{d-1}\partial_kh\otimes e_k,
&
\mathsf L_{\uparrow}h
&:=
\sum_{k=0}^{d-1}\partial_k^*h\otimes e_k,
\\
\mathsf L^-h
&:=
\sum_{k=0}^{d-1}N_kh\otimes e_k,
&
\mathsf L^+h
&:=
\sum_{k=0}^{d-1}\overline N_kh\otimes e_k,
\end{aligned}
\qquad
h\in\mathfrak h_{d-1},
\label{eq:qbn-coin-lifts}
\end{equation}
where
\[
N_k:=\partial_k^*\partial_k,
\qquad
\overline N_k:=\partial_k\partial_k^*=I-N_k.
\]

The lifts
\(\mathsf L_{\downarrow}\) and \(\mathsf L_{\uparrow}\)
map each Hamming layer to the adjacent lower and upper layers, whereas
\(\mathsf L^-\) and \(\mathsf L^+\) resolve, at the current position,
the coin directions leading to the lower and upper layers.

For \(h\in V_r\),
\[
\sum_{k=0}^{d-1}N_kh=rh,
\qquad
\sum_{k=0}^{d-1}\overline N_kh=(d-r)h.
\]
Moreover,
\begin{equation}
\|\mathsf L_{\downarrow}h\|^2
=
\|\mathsf L^-h\|^2
=
r\|h\|^2,
\qquad
\|\mathsf L_{\uparrow}h\|^2
=
\|\mathsf L^+h\|^2
=
(d-r)\|h\|^2.
\label{eq:qbn-lift-norms}
\end{equation}

The QBN lifts also satisfy the orthogonality relations
\begin{equation}
\begin{aligned}
\langle \mathsf L^-h,\mathsf L^+q\rangle
&=0,
&
\langle \mathsf L^-h,\mathsf L_{\downarrow}q\rangle
&=0,
\\
\langle \mathsf L^+h,\mathsf L_{\uparrow}q\rangle
&=0,
\end{aligned}
\qquad
h,q\in\mathfrak h_{d-1}.
\label{eq:qbn-lift-orthogonality}
\end{equation}
Indeed, these identities follow from
\[
N_k\overline N_k=0,
\qquad
N_k\partial_k=0,
\qquad
\overline N_k\partial_k^*=0.
\]

It is useful to isolate the conditional shift
\[
\mathsf S
:=
\sum_{k=0}^{d-1}\Xi_k\otimes P_k.
\]
The four QBN lifts are related by
\begin{equation}
\mathsf S\mathsf L^-
=
\mathsf L_{\downarrow},
\qquad
\mathsf S\mathsf L^+
=
\mathsf L_{\uparrow},
\qquad
\mathsf S\mathsf L_{\downarrow}
=
\mathsf L^-,
\qquad
\mathsf S\mathsf L_{\uparrow}
=
\mathsf L^+.
\label{eq:qbn-shift-lifts}
\end{equation}
These identities follow directly from
\(
\Xi_kZ_\sigma=Z_{\sigma\triangle\{k\}}.
\)

Fix
\(
\varepsilon\in\{\pm1\}.
\)
If \(d=2m\), then
\begin{equation}
s_\sigma=\varepsilon,
\qquad
|\sigma|\in\{m-1,m,m+1\},
\label{eq:even-supporting-spin}
\end{equation}
while the spin values outside these three supporting layers are
arbitrary.  If \(d=2m+1\), then
\begin{equation}
s_\sigma=\varepsilon,
\qquad
|\sigma|\in\{m-1,m,m+1,m+2\},
\label{eq:odd-supporting-spin}
\end{equation}
while the spin values outside these four supporting layers are
arbitrary.  These assumptions will be used throughout the remainder of
the paper.

For the common spin value \(\varepsilon\), set
\[
\omega_\varepsilon
:=
e^{-i\varepsilon\pi/d},
\qquad
\Theta_\varepsilon
:=
\omega_\varepsilon^{\mathcal N}.
\]
Thus \(\Theta_\varepsilon\) acts on \(V_r\) as multiplication by
\(\omega_\varepsilon^r\).  Let
\[
\mathsf G
:=
I_{\mathfrak h_{d-1}}\otimes G.
\]
On states supported on layers carrying the common spin value
\(\varepsilon\), the reduced walk is represented by
\begin{equation}
\widetilde W_{\mathbf s}
=
\mathsf S
(\Theta_\varepsilon\otimes I_{\mathcal K})
\mathsf G.
\label{eq:qbn-factorized-walk}
\end{equation}

%%%%%%%%%%%%%%%%%%%%%%%%%%%%%%%%%%%%%%%%%%%%%%%%%%%%%%%%%%%%%%%%%%%
\subsection{Even-dimensional fixed subspaces}
\label{sec:three-layer-fixed-point}
%%%%%%%%%%%%%%%%%%%%%%%%%%%%%%%%%%%%%%%%%%%%%%%%%%%%%%%%%%%%%%%%%%%

Let \(d=2m\) and let
\(
f\in\mathscr H_m^-.
\)
By Theorem~\ref{thm:harmonic-qbn} we have
\(
A^-f=A^+f=0.
\)

Define the QBN three-layer lift
\begin{equation}
\Phi^{(\varepsilon)}
:
\mathscr H_m^-
\longrightarrow
\left(
V_{m-1}\oplus V_m\oplus V_{m+1}
\right)\otimes\mathcal K
%\label{eq:three-layer-lift-map}
\end{equation}
by
\begin{equation}
\Phi^{(\varepsilon)}f
=
\Phi_f^{(\varepsilon)}
:=
\omega_\varepsilon^{m+1}\mathsf L_{\downarrow}f
+
\bigl(
\mathsf L^-
+
\omega_\varepsilon\mathsf L^+
\bigr)f
+
\omega_\varepsilon^m\mathsf L_{\uparrow}f,
\qquad
f\in\mathscr H_m^-.
\label{eq:three-layer-state}
\end{equation}
Here
\(
\mathsf L_{\downarrow}f
\in
V_{m-1}\otimes\mathcal K,\quad
\mathsf L^-f,\,
\mathsf L^+f
\in
V_m\otimes\mathcal K,
\)
and
\(
\mathsf L_{\uparrow}f
\in
V_{m+1}\otimes\mathcal K.
\)
Hence \(\Phi_f^{(\varepsilon)}\) is supported on the three central
Hamming layers.

Since
\[
A^-f=A^+f=0,
\]
the sums of the coin components of
\(\mathsf L_{\downarrow}f\) and
\(\mathsf L_{\uparrow}f\) vanish.  Hence the rank-one part of the
Grover coin vanishes on these components, and
\begin{equation}
\mathsf G\mathsf L_{\downarrow}f
=
-\mathsf L_{\downarrow}f,
\qquad
\mathsf G\mathsf L_{\uparrow}f
=
-\mathsf L_{\uparrow}f.
\label{eq:Grover-adjacent}
\end{equation}

For the central component, since \(f\in V_m\) and \(d=2m\), the sum
of the coin components of
\(
\bigl(
\mathsf L^-
+
\omega_\varepsilon\mathsf L^+
\bigr)f
\)
is
\(
m(1+\omega_\varepsilon)f.
\)
It follows that
\begin{equation}
\mathsf G
\bigl(
\mathsf L^-
+
\omega_\varepsilon\mathsf L^+
\bigr)f
=
\bigl(
\omega_\varepsilon\mathsf L^-
+
\mathsf L^+
\bigr)f.
\label{eq:Gu-middle}
\end{equation}

Define
\begin{equation}
\mathcal H_{\mathrm{har}}^{\mathrm{even}}
:=
\operatorname{Ran}\Phi^{(\varepsilon)}.
\label{eq:even-harmonic-subspace}
\end{equation}

\begin{theorem}
\label{thm:even-harmonic-subspace}
Let \(d=2m\).  The QBN lift
\[
\Phi^{(\varepsilon)}
:
\mathscr H_m^-
\longrightarrow
\mathcal H_{\mathrm{har}}^{\mathrm{even}}
\]
is a linear isomorphism, and
\begin{equation}
\widetilde W_{\mathbf s}
\big|_{\mathcal H_{\mathrm{har}}^{\mathrm{even}}}
=
I_{\mathcal H_{\mathrm{har}}^{\mathrm{even}}}.
\label{eq:even-subspace-identity}
\end{equation}
Moreover,
\begin{equation}
\dim\mathcal H_{\mathrm{har}}^{\mathrm{even}}
=
\dim\mathscr H_m^-
=
C_m.
\label{eq:even-harmonic-subspace-dimension}
\end{equation}
\end{theorem}

\begin{proof}
Apply the factorization \eqref{eq:qbn-factorized-walk} to
\eqref{eq:three-layer-state}.  Using
\eqref{eq:Grover-adjacent} and \eqref{eq:Gu-middle}, the phase
operator produces
\(
-\omega_\varepsilon^{2m}\mathsf L_{\downarrow}f,\quad
\omega_\varepsilon^m
\bigl(
\omega_\varepsilon\mathsf L^-
+
\mathsf L^+
\bigr)f,
\quad
-\omega_\varepsilon^{2m+1}\mathsf L_{\uparrow}f.
\)

Since
\(
\omega_\varepsilon^{2m}=-1,
\)
applying the conditional shift and using
\eqref{eq:qbn-shift-lifts} gives
\[
\widetilde W_{\mathbf s}\Phi_f^{(\varepsilon)}
=
\Phi_f^{(\varepsilon)}.
\]
Hence every vector in
\(\mathcal H_{\mathrm{har}}^{\mathrm{even}}\) is fixed.

If
\(
\Phi_f^{(\varepsilon)}=0,
\)
then the two central components
\(\mathsf L^-f\) and \(\mathsf L^+f\), which have disjoint coin
supports, both vanish. By \eqref{eq:qbn-lift-norms} we have \(f=0\).  Thus \(\Phi^{(\varepsilon)}\) is injective and hence an
isomorphism onto its range.  The dimension formula follows from
Corollary~\ref{thm:dim-harmonic-qbn}.
\end{proof}

Equivalently, the QBN lift intertwines the identity dynamics on the
central harmonic subspace with the walk:
\begin{equation}
\widetilde W_{\mathbf s}\Phi^{(\varepsilon)}
=
\Phi^{(\varepsilon)}
I_{\mathscr H_m^-}.
\label{eq:even-intertwining}
\end{equation}

The next result shows that this QBN lift exhausts all fixed points with
the same three-layer support.

\begin{theorem}
\label{thm:even-harmonic-completeness}
Let \(d=2m\) with \(m\ge2\).  Set
\[
\mathscr L_m^{\mathrm{even}}
:=
\left(
V_{m-1}\oplus V_m\oplus V_{m+1}
\right)\otimes\mathcal K.
\]
Then
\begin{equation}
\left\{
\Phi\in\mathscr L_m^{\mathrm{even}}:
\widetilde W_{\mathbf s}\Phi=\Phi
\right\}
=
\mathcal H_{\mathrm{har}}^{\mathrm{even}}.
\label{eq:even-harmonic-completeness}
\end{equation}
\end{theorem}

\begin{proof}
The inclusion
\[
\mathcal H_{\mathrm{har}}^{\mathrm{even}}
\subseteq
\left\{
\Phi\in\mathscr L_m^{\mathrm{even}}:
\widetilde W_{\mathbf s}\Phi=\Phi
\right\}
\]
follows from Theorem~\ref{thm:even-harmonic-subspace}.  We prove the
reverse inclusion.  By complex conjugation, it is enough to consider
\(\varepsilon=1\).

Let
\[
\Phi
=
\Phi_{m-1}+\Phi_m+\Phi_{m+1}
\in
\mathscr L_m^{\mathrm{even}}
\]
satisfy
\[
\widetilde W_{\mathbf s}\Phi=\Phi,
\]
where
\(
\Phi_r\in V_r\otimes\mathcal K.
\)
Since \(\Phi\) has no components in
\(V_{m-2}\otimes\mathcal K\) or
\(V_{m+2}\otimes\mathcal K\), the fixed-point equation implies that
the outward components created from the two outer layers vanish.
Equivalently, the sums of the coin amplitudes on
\(\Phi_{m-1}\) and \(\Phi_{m+1}\) vanish.  Hence the Grover coin acts
as \(-I\) on these two components.

The fixed-point equation between the three remaining layers can
therefore be solved successively from the middle layer.  There is a
unique \(f\in V_m\) such that
\begin{equation}
\Phi_m
=
\bigl(
\mathsf L^-
+
\omega\mathsf L^+
\bigr)f.
\label{eq:even-completeness-middle-form}
\end{equation}
The equations on the adjacent layers then give
\begin{equation}
\Phi_{m-1}
=
\omega^{m+1}\mathsf L_{\downarrow}f,
\qquad
\Phi_{m+1}
=
\omega^m\mathsf L_{\uparrow}f.
\label{eq:even-completeness-outer-form}
\end{equation}

Because the sums of the coin components of
\(\Phi_{m-1}\) and \(\Phi_{m+1}\) vanish, 
\eqref{eq:even-completeness-outer-form} yields
\[
A^-f=0,
\qquad
A^+f=0.
\]
Thus
\[
f\in
\mathscr H_m^-\cap\mathscr H_m^+
=
\mathscr H_m^-,
\]
where the last equality follows from
Theorem~\ref{thm:harmonic-qbn}.  Combining
\eqref{eq:even-completeness-middle-form} and
\eqref{eq:even-completeness-outer-form}, we obtain
\[
\Phi
=
\omega^{m+1}\mathsf L_{\downarrow}f
+
\bigl(
\mathsf L^-+\omega\mathsf L^+
\bigr)f
+
\omega^m\mathsf L_{\uparrow}f
=
\Phi_f^{(1)}.
\]
Hence
\(
\Phi\in
\mathcal H_{\mathrm{har}}^{\mathrm{even}}.
\)
The case \(\varepsilon=-1\) follows by complex conjugation.
\end{proof}

\begin{remark}
The hypothesis \(m\ge2\) in
Theorem~\ref{thm:even-harmonic-completeness} is necessary.  When
\(d=2\), the three supporting layers comprise the whole hypercube.
A direct calculation for \(\varepsilon=1\) gives a two-dimensional
fixed-point space, whereas
\[
\dim\mathscr H_1^-=C_1=1.
\]
The case \(\varepsilon=-1\) follows by complex conjugation.  Thus the
QBN lift remains a fixed subspace but is not complete in this
exceptional case.
\end{remark}

%%%%%%%%%%%%%%%%%%%%%%%%%%%%%%%%%%%%%%%%%%%%%%%%%%%%%%%%%%%%%%%%%%%
\subsection{Odd-dimensional harmonic reduction}
\label{sec:odd-invariant-subspace}
%%%%%%%%%%%%%%%%%%%%%%%%%%%%%%%%%%%%%%%%%%%%%%%%%%%%%%%%%%%%%%%%%%%

Let \(d=2m+1\), \(m\ge1\).  For
\(
f\in\mathscr H_m^-,
\)
set
\(
g:=A^+f.
\)
By Theorem~\ref{thm:harmonic-qbn},
\[
A^-f=0,
\qquad
A^+f=g,
\qquad
A^-g=f,
\qquad
A^+g=0,
\]
with
\[
f\in V_m,
\qquad
g\in\mathscr H_{m+1}^+\subset V_{m+1}.
\]
These two vectors give rise to eight canonical QBN
position--coin components, distributed over the four central Hamming
layers as
\[
\mathsf L_{\downarrow}f
\in
V_{m-1}\otimes\mathcal K,\qquad
\mathsf L^-f,\quad
\mathsf L^+f,\quad
\mathsf L_{\downarrow}g
\in
V_m\otimes\mathcal K,
\]
\[
\mathsf L_{\uparrow}f,\quad
\mathsf L^-g,\quad
\mathsf L^+g
\in
V_{m+1}\otimes\mathcal K,\qquad
\mathsf L_{\uparrow}g
\in
V_{m+2}\otimes\mathcal K.
\]
Accordingly, define the four-layer position--coin space
\[
\mathscr L_m^{\mathrm{odd}}
:=
\left(
V_{m-1}\oplus V_m\oplus V_{m+1}\oplus V_{m+2}
\right)\otimes\mathcal K.
\]
For
\(
x=(x_0,\ldots,x_7)^T\in\mathbb C^8,
\)
set
\[
\begin{aligned}
\Psi_f(x)
:={}&
x_0\mathsf L_{\downarrow}f
+x_1\mathsf L^-f
+x_2\mathsf L^+f
+x_3\mathsf L_{\downarrow}g
\\
&+
x_4\mathsf L_{\uparrow}f
+x_5\mathsf L^-g
+x_6\mathsf L^+g
+x_7\mathsf L_{\uparrow}g.
\end{aligned}
\]
Since \(\Psi_f(x)\) is bilinear in \(f\) and \(x\), it induces the
linear QBN four-layer lift
\[
\mathfrak T_m:
\mathscr H_m^-\otimes\mathbb C^8
\longrightarrow
\mathscr L_m^{\mathrm{odd}},
\qquad
\mathfrak T_m(f\otimes x)
=
\Psi_f(x).
\]

We next determine the action of the Grover coin on these components.
Recall that
\[
\mathsf G
=
I_{\mathfrak h_{d-1}}\otimes G,
\qquad
G
=
2|u_d\rangle\langle u_d|-I_{\mathcal K},
\]
where
\(
u_d
=
\frac1{\sqrt d}
\sum_{k=0}^{d-1}e_k.
\)
For
\[
v
=
\sum_{k=0}^{d-1}h_k\otimes e_k
\in
\mathfrak h_{d-1}\otimes\mathcal K,
\]
one has
\[
\mathsf Gv
=
\frac{2}{d}
\sum_{j=0}^{d-1}
\left(
\sum_{k=0}^{d-1}h_k
\right)\otimes e_j
-v.
\]

Since \(A^-f=0\) and \(A^+g=0\),
\[
\mathsf G\mathsf L_{\downarrow}f
=
-\mathsf L_{\downarrow}f,
\qquad
\mathsf G\mathsf L_{\uparrow}g
=
-\mathsf L_{\uparrow}g.
\]
For the two central components,
\[
\sum_{k=0}^{d-1}
\left(
x_1N_kf
+x_2\overline N_kf
+x_3\partial_kg
\right)
=
\bigl[
mx_1+(m+1)x_2+x_3
\bigr]f,
\]
and
\[
\sum_{k=0}^{d-1}
\left(
x_4\partial_k^*f
+x_5N_kg
+x_6\overline N_kg
\right)
=
\bigl[
x_4+(m+1)x_5+mx_6
\bigr]g.
\]
Combining these identities with
\eqref{eq:qbn-shift-lifts} and
\eqref{eq:qbn-factorized-walk}, one obtains
\begin{equation}
\widetilde W_{\mathbf s}\Psi_f(x)
=
\Psi_f\bigl(\mathcal M_m^{(\varepsilon)}x\bigr).
\label{eq:odd-reduced-intertwining-vector}
\end{equation}
Writing
\[
\mathcal M_m^{(\varepsilon)}x
=
(y_0,\ldots,y_7)^T,
\]
we have
\begin{equation}
\begin{aligned}
y_0
&=
\frac{\omega_\varepsilon^m}{d}
\bigl[
-x_1+2(m+1)x_2+2x_3
\bigr],
&
y_1
&=
-\omega_\varepsilon^{m-1}x_0,
\\
y_2
&=
-\omega_\varepsilon^{m+1}x_4,
&
y_3
&=
\frac{\omega_\varepsilon^{m+1}}{d}
\bigl[
2x_4+x_5+2mx_6
\bigr],
\\
y_4
&=
\frac{\omega_\varepsilon^m}{d}
\bigl[
2mx_1+x_2+2x_3
\bigr],
&
y_5
&=
-\omega_\varepsilon^m x_3,
\\
y_6
&=
-\omega_\varepsilon^{m+2}x_7,
&
y_7
&=
\frac{\omega_\varepsilon^{m+1}}{d}
\bigl[
2x_4+2(m+1)x_5-x_6
\bigr].
\end{aligned}
\label{eq:odd-reduced-map}
\end{equation}

Since
\(
\omega_{-1}
=
\overline{\omega_{+1}},
\)
we also have
\(
\mathcal M_m^{(-1)}
=
\overline{\mathcal M_m^{(+1)}}.
\)
Define the odd-dimensional QBN harmonic subspace by
\begin{equation*}
\mathcal H_{\mathrm{har}}^{\mathrm{odd}}
:=
\operatorname{Ran}\mathfrak T_m
\subset
\mathscr L_m^{\mathrm{odd}}.
\end{equation*}

\begin{theorem}
\label{thm:odd-harmonic-subspace}
Let \(d=2m+1\), \(m\ge1\).  Then
\[
\mathfrak T_m:
\mathscr H_m^-\otimes\mathbb C^8
\longrightarrow
\mathcal H_{\mathrm{har}}^{\mathrm{odd}}
\]
is a linear isomorphism.  Moreover,
\begin{equation}
\widetilde W_{\mathbf s}\mathfrak T_m
=
\mathfrak T_m
\left(
I_{\mathscr H_m^-}
\otimes
\mathcal M_m^{(\varepsilon)}
\right).
\label{eq:odd-global-intertwining}
\end{equation}
Consequently,
\begin{equation}
\widetilde W_{\mathbf s}
\big|_{\mathcal H_{\mathrm{har}}^{\mathrm{odd}}}
\sim
I_{\mathscr H_m^-}
\otimes
\mathcal M_m^{(\varepsilon)},
\label{eq:odd-subspace-similarity}
\end{equation}
and
\begin{equation}
\dim\mathcal H_{\mathrm{har}}^{\mathrm{odd}}
=
8\dim\mathscr H_m^-
=
8C_{m+1}.
\label{eq:odd-harmonic-subspace-dimension}
\end{equation}
\end{theorem}

\begin{proof}
It remains to prove that \(\mathfrak T_m\) is injective.  We use the
QBN lift identities
\[
(\mathsf L^+)^*\mathsf L_{\downarrow}=A^-,
\qquad
\mathsf L_{\downarrow}^*\mathsf L^+=A^+,
\]
and
\[
(\mathsf L^-)^*\mathsf L_{\uparrow}=A^+,
\qquad
\mathsf L_{\uparrow}^*\mathsf L^-=A^-.
\]
Identify
\[
\mathscr H_m^-\otimes\mathbb C^8
\cong
(\mathscr H_m^-)^8,
\]
and suppose
\[
\mathfrak T_m(f_0,\ldots,f_7)=0,
\qquad
f_j\in\mathscr H_m^-.
\]
Set
\(
g_j:=A^+f_j.
\)

The components on the two outer layers give
\[
\mathsf L_{\downarrow}f_0=0,
\qquad
\mathsf L_{\uparrow}g_7=0.
\]
By \eqref{eq:qbn-lift-norms} and the isometry
\(
A^+:\mathscr H_m^-\to\mathscr H_{m+1}^+,
\)
we obtain
\[
f_0=f_7=0.
\]

On \(V_m\otimes\mathcal K\), the QBN lift orthogonality
\eqref{eq:qbn-lift-orthogonality} gives
\[
\mathsf L^-f_1=0,
\qquad
\mathsf L^+f_2+\mathsf L_{\downarrow}g_3=0.
\]
By \eqref{eq:qbn-lift-norms}, the first identity gives \(f_1=0\).
Applying \((\mathsf L^+)^*\) and
\(\mathsf L_{\downarrow}^*\) to the second identity gives
\[
(m+1)f_2+f_3=0,
\qquad
g_2+(m+1)g_3=0.
\]
Since
\[
g_j=A^+f_j,
\qquad
A^-g_j=f_j,
\]
the second relation yields
\[
f_2+(m+1)f_3=0.
\]
Since
\[
(m+1)^2-1=m(m+2)\ne0,
\]
we obtain
\[
f_2=f_3=0.
\]

On \(V_{m+1}\otimes\mathcal K\), the same QBN lift orthogonality gives
\[
\mathsf L^+g_6=0,
\qquad
\mathsf L_{\uparrow}f_4+\mathsf L^-g_5=0.
\]
By \eqref{eq:qbn-lift-norms} and the isometry of \(A^+\), the first
identity gives \(f_6=0\).  Applying
\((\mathsf L^-)^*\) and \(\mathsf L_{\uparrow}^*\) to the second
identity gives
\[
g_4+(m+1)g_5=0,
\qquad
(m+1)f_4+f_5=0.
\]
Applying \(A^-\) to the first relation yields
\[
f_4+(m+1)f_5=0.
\]
Hence
\[
f_4=f_5=0.
\]

Thus
\[
f_0=\cdots=f_7=0,
\]
so \(\mathfrak T_m\) is injective.

Equation~\eqref{eq:odd-reduced-intertwining-vector} now gives
\eqref{eq:odd-global-intertwining}.  Since \(\mathfrak T_m\) is an
isomorphism onto its range,
\eqref{eq:odd-subspace-similarity} follows.  Finally,
\[
\dim\mathscr H_m^-=C_{m+1},
\]
and hence
\(
\dim\mathcal H_{\mathrm{har}}^{\mathrm{odd}}
=
8C_{m+1}.
\)
\end{proof}

%%%%%%%%%%%%%%%%%%%%%%%%%%%%%%%%%%%%%%%%%%%%%%%%%%%%%%%%%%%%%%%%%%%
% Isotypic completeness
\subsection{Isotypic completeness}
%%%%%%%%%%%%%%%%%%%%%%%%%%%%%%%%%%%%%%%%%%%%%%%%%%%%%%%%%%%%%%%%%%%
We next show that this eight-dimensional factor is exactly the
multiplicity space of the corresponding symmetric-group representation.

Let \(S_d\) be the symmetric group on \(\mathbb N_{d-1}\).  For
\(\pi\in S_d\) and \(\sigma\in\Gamma_{d-1}\), define
\[
\pi(\sigma)
:=
\{\pi(k):k\in\sigma\}.
\]
The simultaneous action of \(S_d\) on the position--coin space is
defined by
\begin{equation}
\mathcal R_\pi
(Z_\sigma\otimes e_k)
:=
Z_{\pi(\sigma)}\otimes e_{\pi(k)}.
\label{eq:simultaneous-permutation-action}
\end{equation}

Let \(S^{(m+1,m)}\) denote the Specht module associated with the
partition \((m+1,m)\) of \(d=2m+1\).  We write
\[
\left(
\mathscr L_m^{\mathrm{odd}}
\right)_{(m+1,m)}
\]
for the corresponding \(S^{(m+1,m)}\)-isotypic component of
\(\mathscr L_m^{\mathrm{odd}}\).

\begin{theorem}
\label{thm:odd-harmonic-completeness}
Let \(d=2m+1\), \(m\ge1\).  Then
\begin{equation}
\mathscr H_m^-
\cong
S^{(m+1,m)}
\label{eq:harmonic-specht}
\end{equation}
as \(S_d\)-modules, and
\begin{equation}
\mathcal H_{\mathrm{har}}^{\mathrm{odd}}
=
\left(
\mathscr L_m^{\mathrm{odd}}
\right)_{(m+1,m)}.
\label{eq:odd-isotypic-completeness}
\end{equation}
Moreover, \(S^{(m+1,m)}\) occurs with multiplicity \(8\) in
\(\mathscr L_m^{\mathrm{odd}}\). 
\end{theorem}

\begin{proof}
Let
\[
\rho_\pi Z_\sigma
:=
Z_{\pi(\sigma)},
\qquad
\pi\in S_d,
\]
denote the induced representation of \(S_d\) on
\(\mathfrak h_{d-1}\).  The QBN operators satisfy
\[
\rho_\pi\partial_k\rho_\pi^{-1}
=
\partial_{\pi(k)},
\qquad
\rho_\pi\partial_k^*\rho_\pi^{-1}
=
\partial_{\pi(k)}^*.
\]
Hence
\[
\rho_\pi A^-
=
A^-\rho_\pi,
\qquad
\rho_\pi A^+
=
A^+\rho_\pi,
\]
so the QBN harmonic subspaces are \(S_d\)-invariant.

The permutation module \(V_r\) has the multiplicity-free decomposition
\begin{equation}
V_r
\cong
\bigoplus_{j=0}^{\min(r,d-r)}
S^{(d-j,j)};
\label{eq:slice-specht-decomposition}
\end{equation}
see \cite{Sagan2001}.  For \(d=2m+1\),
\[
V_m
\cong
\bigoplus_{j=0}^{m}
S^{(2m+1-j,j)},
\]
whereas
\[
V_{m-1}
\cong
\bigoplus_{j=0}^{m-1}
S^{(2m+1-j,j)}.
\]
Thus the only irreducible summand of \(V_m\) that does not occur in
\(V_{m-1}\) is
\(
S^{(m+1,m)}.
\) Since
\(
A^-
\)
is \(S_d\)-equivariant and surjective, its kernel is precisely this
remaining irreducible summand.  Therefore
\begin{equation}
\mathscr H_m^-
=
\ker\left(A^-\big|_{V_m}\right)
\cong
S^{(m+1,m)}.
\label{eq:harmonic-specht-module}
\end{equation}

We next determine the multiplicity of \(S^{(m+1,m)}\) in
\(\mathscr L_m^{\mathrm{odd}}\).  For each \(r\), decompose
\begin{equation}
V_r\otimes\mathcal K
=
\mathscr L_r^-
\oplus
\mathscr L_r^+,
\label{eq:direction-space-decomposition}
\end{equation}
where
\[
\mathscr L_r^-
:=
\operatorname{span}
\left\{
Z_\sigma\otimes e_k:
|\sigma|=r,\ k\in\sigma
\right\},
\quad
\mathscr L_r^+
:=
\operatorname{span}
\left\{
Z_\sigma\otimes e_k:
|\sigma|=r,\ k\notin\sigma
\right\}.
\]
Thus \(\mathscr L_r^-\) consists of the coin directions leading from
the \(r\)-th Hamming layer to the \((r-1)\)-th layer, whereas
\(\mathscr L_r^+\) consists of those leading to the \((r+1)\)-th
layer.

Let \(M_s^{(2m)}\) denote the permutation module of \(S_{2m}\) on the
\(s\)-subsets of a \(2m\)-element set.  Fixing one coin coordinate gives
\begin{equation}
\mathscr L_r^-
\cong
\operatorname{Ind}_{S_{2m}}^{S_{2m+1}}
M_{r-1}^{(2m)},
\qquad
\mathscr L_r^+
\cong
\operatorname{Ind}_{S_{2m}}^{S_{2m+1}}
M_r^{(2m)}.
\label{eq:direction-induced-modules}
\end{equation}
By the branching rule,
\begin{equation}
\operatorname{Res}_{S_{2m}}^{S_{2m+1}}
S^{(m+1,m)}
\cong
S^{(m,m)}
\oplus
S^{(m+1,m-1)}.
\label{eq:specht-branching}
\end{equation}
Using Frobenius reciprocity together with the decomposition of the
permutation modules \(M_s^{(2m)}\), the multiplicities of
\(S^{(m+1,m)}\) are
\[
\begin{array}{c|cc|c}
r
&
\mathscr L_r^-
&
\mathscr L_r^+
&
V_r\otimes\mathcal K
\\ \hline
m-1&0&1&1\\
m&1&2&3\\
m+1&2&1&3\\
m+2&1&0&1.
\end{array}
\]
Therefore \(S^{(m+1,m)}\) occurs in
\(\mathscr L_m^{\mathrm{odd}}\) with multiplicity
\(
1+3+3+1=8.
\)
Consequently,
\begin{equation}
\dim
\left(
\mathscr L_m^{\mathrm{odd}}
\right)_{(m+1,m)}
=
8\dim S^{(m+1,m)}
\label{eq:odd-isotypic-dimension-proof}
\end{equation}

Moreover, the QBN lifts satisfy
\[
\mathcal R_\pi\mathsf L_{\downarrow}
=
\mathsf L_{\downarrow}\rho_\pi,
\qquad
\mathcal R_\pi\mathsf L_{\uparrow}
=
\mathsf L_{\uparrow}\rho_\pi,\qquad
\mathcal R_\pi\mathsf L^-
=
\mathsf L^-\rho_\pi,
\qquad
\mathcal R_\pi\mathsf L^+
=
\mathsf L^+\rho_\pi.
\]
Since \(A^+\) is \(S_d\)-equivariant, the four-layer lift satisfies
\begin{equation}
\mathcal R_\pi\mathfrak T_m
=
\mathfrak T_m
\left(
\rho_\pi\big|_{\mathscr H_m^-}
\otimes
I_{\mathbb C^8}
\right).
\label{eq:Tm-equivariance}
\end{equation}
Therefore
\(
\mathcal H_{\mathrm{har}}^{\mathrm{odd}}
=
\operatorname{Ran}\mathfrak T_m
\)
is an \(S_d\)-submodule.  Since
\(
\mathfrak T_m
\)
is an \(S_d\)-equivariant isomorphism and \(S_d\) acts trivially on
\(\mathbb C^8\),
\[
\mathcal H_{\mathrm{har}}^{\mathrm{odd}}
\cong
\left(
S^{(m+1,m)}
\right)^{\oplus8}.
\]
Hence
\[
\mathcal H_{\mathrm{har}}^{\mathrm{odd}}
\subseteq
\left(
\mathscr L_m^{\mathrm{odd}}
\right)_{(m+1,m)}.
\]
By Theorem~\ref{thm:odd-harmonic-subspace}
and \eqref{eq:odd-isotypic-dimension-proof} we have
\[
\mathcal H_{\mathrm{har}}^{\mathrm{odd}}
=
\left(
\mathscr L_m^{\mathrm{odd}}
\right)_{(m+1,m)}.
\]
\end{proof}

\subsection{Spectrum of the harmonic dynamics}
\label{sec:harmonic-spectrum}

We now determine
the spectrum of the nontrivial odd-dimensional dynamics.

Let
\[
\chi_m(\lambda)
:=
\det\!\left(
\lambda I-\mathcal M_m^{(\varepsilon)}
\right).
\]
A direct calculation from \eqref{eq:odd-reduced-map} gives
\begin{equation*}
%\label{eq:odd-characteristic-polynomial}
\chi_m(\lambda)
=
\lambda^8+A_d\lambda^6+B_d\lambda^4+A_d\lambda^2+1,
\end{equation*}
where
\begin{equation*}
%\label{eq:odd-spectrum-coefficients}
A_d
=
\frac4{d^2}
\left(
1-d\sin^2\frac{\pi}{d}
\right),
\qquad
B_d
=
\frac{2+2(3-d^2)\cos(2\pi/d)}{d^2}.
\end{equation*}
The same polynomial applies for both \(\varepsilon=\pm1\), since
\(
\mathcal M_m^{(-1)}
=
\overline{\mathcal M^{(+1)}_m}
\)
and the coefficients of \(\chi_m\) are real.

Write
\[
\mathbb T:=\{z\in\mathbb C:|z|=1\}.
\]
By Theorem~\ref{thm:odd-harmonic-subspace},
\[
\widetilde W_{\mathbf s}
\big|_{\mathcal H_{\mathrm{har}}^{\mathrm{odd}}}
\sim
I_{\mathscr H_m^-}\otimes
\mathcal M_m^{(\varepsilon)}.
\]
Since the restriction on the left-hand side is unitary,
\begin{equation}
\label{eq:odd-reduced-spectrum-unit-circle}
\operatorname{Spec}
\bigl(
\mathcal M_m^{(\varepsilon)}
\bigr)
\subseteq
\mathbb T.
\end{equation}

Let \(\lambda\) be a root of \(\chi_m\). Since
\(\chi_m(0)=1\), we have \(\lambda\neq0\). Dividing
\(\chi_m(\lambda)=0\) by \(\lambda^4\) and setting
\(
y=\lambda^2+\lambda^{-2},
\)
we obtain
\[
y^2+A_dy+(B_d-2)=0.
\]
For \(d=2m+1\ge3\), its discriminant is
\[
\Delta_y
=
\frac{16}{d^4}
\left[
d^4\cos^2\frac{\pi}{d}
-2d^2
+d^2s^2
+d(3d-2)s
+1
\right],
\qquad
s=\sin^2\frac{\pi}{d}.
\]
Since
\(
\cos^2\frac{\pi}{d}\ge\frac14
\)
and
\(
d^2\left(\frac{d^2}{4}-2\right)>0
\qquad(d\ge3),
\)
we have \(\Delta_y>0\). Thus the quadratic has two distinct real
roots \(y_+\) and \(y_-\), and
\begin{equation*}
%\label{eq:odd-characteristic-factorization}
\chi_m(\lambda)
=
\bigl(\lambda^4-y_+\lambda^2+1\bigr)
\bigl(\lambda^4-y_-\lambda^2+1\bigr).
\end{equation*}

Every root of either quartic factor is a root of \(\chi_m\), and hence
lies on \(\mathbb T\) by
\eqref{eq:odd-reduced-spectrum-unit-circle}. Therefore
\[
y_\pm
=
\lambda^2+\lambda^{-2}
=
2\operatorname{Re}(\lambda^2)
\in[-2,2].
\]
Moreover,
\[
\chi_m(\pm1)
=
\frac4{d^2}
\left[
4+(d-3)(d+1)
\sin^2\frac{\pi}{d}
\right]
>0,
\]
while
\[
\chi_m(\pm i)
=
\frac{4(d-1)(d+3)}{d^2}
\sin^2\frac{\pi}{d}
>0.
\]
Since \(y=2\) corresponds to \(\lambda=\pm1\) and
\(y=-2\) to \(\lambda=\pm i\), it follows that
\[
y_+,y_-\in(-2,2).
\]

For each \(y\in\{y_+,y_-\}\), setting \(z=\lambda^2\) reduces
\(
\lambda^4-y\lambda^2+1=0
\)
to
\(
z^2-yz+1=0.
\)
Since \(y\in(-2,2)\), this equation has two distinct conjugate roots
on \(\mathbb T\), and hence the corresponding quartic has four
distinct roots on \(\mathbb T\). Since \(y_+\neq y_-\), the two
quartic factors have no common root. Therefore \(\chi_m\) has eight
distinct roots on \(\mathbb T\).

In particular,
\begin{equation}
\label{eq:no-fourth-root-spectrum}
\{1,-1,i,-i\}
\cap
\operatorname{Spec}
\bigl(
\mathcal M_m^{(\varepsilon)}
\bigr)
=
\varnothing.
\end{equation}

\begin{corollary}
\label{cor:harmonic-spectral-classification}
If \(d=2m\), then
\[
\operatorname{Spec}\!\left(
\widetilde W_{\mathbf s}
\big|_{\mathcal H_{\mathrm{har}}^{\mathrm{even}}}
\right)
=
\{1\},
\]
with multiplicity \(C_m\).

If \(d=2m+1\), then
\[
\operatorname{Spec}\!\left(
\widetilde W_{\mathbf s}
\big|_{\mathcal H_{\mathrm{har}}^{\mathrm{odd}}}
\right)
=
\operatorname{Spec}
\bigl(
\mathcal M_m^{(\varepsilon)}
\bigr),
\]
which consists of eight distinct eigenvalues on \(\mathbb T\), each
with multiplicity exactly \(C_{m+1}\).
\end{corollary}

%\begin{proof}
%For \(d=2m\), the conclusion follows from
%Theorem~\ref{thm:even-harmonic-subspace}, which gives
%\[
%\widetilde W_{\mathbf s}
%\big|_{\mathcal H_{\mathrm{har}}^{\mathrm{even}}}
%=
%I
%\]
%and
%\[
%\dim
%\mathcal H_{\mathrm{har}}^{\mathrm{even}}
%=
%C_m.
%\]

%For \(d=2m+1\),
%Theorem~\ref{thm:odd-harmonic-subspace} gives
%\[
%\widetilde W_{\mathbf s}
%\big|_{\mathcal H_{\mathrm{har}}^{\mathrm{odd}}}
%\sim
%I_{\mathscr H_m^-}
%\otimes
%\mathcal M_m^{(\varepsilon)}.
%\]
%Since
%\[
%\dim\mathscr H_m^-=C_{m+1},
%\]
%and the eight eigenvalues of
%\(\mathcal M_m^{(\varepsilon)}\) are distinct, each occurs with
%multiplicity exactly \(C_{m+1}\) in the harmonic subspace.
%\end{proof}

We conclude with the unified parity-dependent form of the harmonic
dynamics.  Define
\[
\mathcal H_{\mathrm{har}}
:=
\begin{cases}
\mathcal H_{\mathrm{har}}^{\mathrm{even}},
& d=2m,\\[1mm]
\mathcal H_{\mathrm{har}}^{\mathrm{odd}},
& d=2m+1.
\end{cases}
\]
Under \eqref{eq:even-supporting-spin} when \(d=2m\) and
\eqref{eq:odd-supporting-spin} when \(d=2m+1\),
\begin{equation}
\widetilde W_{\mathbf s}
\big|_{\mathcal H_{\mathrm{har}}}
\sim
\begin{cases}
I_{\mathscr H_m^-},
& d=2m,\\[1mm]
I_{\mathscr H_m^-}
\otimes
\mathcal M_m^{(\varepsilon)},
& d=2m+1.
\end{cases}
\label{eq:parity-restricted-dynamics}
\end{equation}
Hence \(\mathcal H_{\mathrm{har}}\) is invariant under
\(\widetilde W_{\mathbf s}\).  In particular, for every
\(\Phi_0\in\mathcal H_{\mathrm{har}}\) and \(t\in\mathbb N\),
\begin{equation}
\widetilde W_{\mathbf s}^{\,t}\Phi_0
\in
\begin{cases}
\left(
V_{m-1}\oplus V_m\oplus V_{m+1}
\right)\otimes\mathcal K,
& d=2m,\\[1mm]
\left(
V_{m-1}\oplus V_m\oplus V_{m+1}\oplus V_{m+2}
\right)\otimes\mathcal K,
& d=2m+1.
\end{cases}
\label{eq:harmonic-layer-confinement}
\end{equation}

%%%%%%%%%%%%%%%%%%%%%%%%%%%%%%%%%%%%%%%%%%%%%%%%%%%%%%%%%%%%%%%%%%%
\section{Long-time behavior and spatial extension of harmonic states}
\label{sec:long-time-geometry}
%%%%%%%%%%%%%%%%%%%%%%%%%%%%%%%%%%%%%%%%%%%%%%%%%%%%%%%%%%%%%%%%%%%

In this section, we examine the walk from the perspective of
probability distributions.  We first determine the time-averaged
position distribution on the harmonic subspace.  We then show that an
arbitrary initial state with nonzero harmonic component retains a
quantitative long-time return memory.  Finally, we compare the harmonic
position distributions with the uniform measure on the hypercube and
show that their confinement produces a persistent obstruction to uniform
mixing.  This non-mixing behavior will subsequently be compared with
the spatial extension of harmonic states in Hamming distance.

%%%%%%%%%%%%%%%%%%%%%%%%%%%%%%%%%%%%%%%%%%%%%%%%%%%%%%%%%%%%%%%%%%%
\subsection{Time-averaged probability distribution}
\label{sec:averaged-probability}
%%%%%%%%%%%%%%%%%%%%%%%%%%%%%%%%%%%%%%%%%%%%%%%%%%%%%%%%%%%%%%%%%%%

Let
\(
\Phi_0\in\mathcal H_{\mathrm{har}}
\)
be a unit vector and write
\begin{equation*}
\Phi_t
:=
\widetilde W_{\mathbf s}^{\,t}\Phi_0,
\qquad
t\ge0.
%\label{eq:harmonic-time-state}
\end{equation*}
For \(\sigma\in\Gamma_{d-1}\), define
\begin{equation*}
P_t(\sigma\mid\Phi_0)
:=
\sum_{k=0}^{d-1}
\left|
\left\langle
Z_\sigma\otimes e_k,
\Phi_t
\right\rangle
\right|^2.
%\label{eq:harmonic-position-probability}
\end{equation*}
Thus
\(
\sigma\longmapsto P_t(\sigma\mid\Phi_0)
\)
is the position probability distribution of the walker at time \(t\).
For \(T\ge1\), define the time-averaged position distribution by
\begin{equation}
P_T(\sigma\mid\Phi_0)
:=
\frac1T
\sum_{t=0}^{T-1}
P_t(\sigma\mid\Phi_0),
\qquad
\sigma\in\Gamma_{d-1}.
\label{eq:harmonic-averaged-probability}
\end{equation}

By Corollary~\ref{cor:harmonic-spectral-classification}, when
\(d=2m+1\), let
\(
\lambda_1,\ldots,\lambda_8
\)
denote the eight distinct eigenvalues of
\(
\widetilde W_{\mathbf s}
\big|_{\mathcal H_{\mathrm{har}}^{\mathrm{odd}}},
\)
and let
\(
\mathsf E_1,\ldots,\mathsf E_8
\)
be the corresponding orthogonal spectral projections.  

The next theorem gives the limiting time-averaged distribution on the
QBN harmonic subspace.

\begin{theorem}
\label{thm:harmonic-averaged-probability}
If \(d=2m\), then for every unit vector
\(
\Phi_0\in\mathcal H_{\mathrm{har}}^{\mathrm{even}},
\)
we have
\begin{equation}
P_t(\sigma\mid\Phi_0)
=
P_0(\sigma\mid\Phi_0),
\qquad
\sigma\in\Gamma_{d-1},\quad t\ge0.
\label{eq:even-probability-stationary}
\end{equation}

If \(d=2m+1\), then, for every unit vector
\(
\Phi_0\in\mathcal H_{\mathrm{har}}^{\mathrm{odd}},
\)
and every \(\sigma\in\Gamma_{d-1}\),
\begin{equation}
\lim_{T\to\infty}
P_T(\sigma\mid\Phi_0)
=
\sum_{a=1}^{8}
\sum_{k=0}^{d-1}
\left|
\left\langle
Z_\sigma\otimes e_k,
\mathsf E_a\Phi_0
\right\rangle
\right|^2.
\label{eq:odd-averaged-probability-limit}
\end{equation}
\end{theorem}

\begin{proof}
The even-dimensional statement follows immediately from
Theorem~\ref{thm:even-harmonic-subspace}. Now let \(d=2m+1\). 
\[
\widetilde W_{\mathbf s}^{\,t}\Phi_0
=
\sum_{a=1}^{8}
\lambda_a^t\mathsf E_a\Phi_0.
\]
For fixed \(\sigma\in\Gamma_{d-1}\), define
\[
v_a(\sigma)
:=
\sum_{k=0}^{d-1}
\left\langle
Z_\sigma\otimes e_k,
\mathsf E_a\Phi_0
\right\rangle e_k
\in\mathcal K.
\]
Then
\[
P_t(\sigma\mid\Phi_0)
=
\left\|
\sum_{a=1}^{8}
\lambda_a^t v_a(\sigma)
\right\|^2,
\]
and hence
\[
P_t(\sigma\mid\Phi_0)
=
\sum_{a,b=1}^{8}
(\lambda_a\overline{\lambda_b})^t
\left\langle
v_b(\sigma),v_a(\sigma)
\right\rangle.
\]
Since the eigenvalues are distinct,
\[
\lim_{T\to\infty}
\frac1T
\sum_{t=0}^{T-1}
(\lambda_a\overline{\lambda_b})^t
=
\begin{cases}
1,&a=b,\\
0,&a\ne b.
\end{cases}
\]
Therefore
\[
\lim_{T\to\infty}
P_T(\sigma\mid\Phi_0)
=
\sum_{a=1}^{8}
\|v_a(\sigma)\|^2.
\]
Since
\[
\|v_a(\sigma)\|^2
=
\sum_{k=0}^{d-1}
\left|
\left\langle
Z_\sigma\otimes e_k,
\mathsf E_a\Phi_0
\right\rangle
\right|^2,
\]
this gives
\eqref{eq:odd-averaged-probability-limit}.
\end{proof}

%%%%%%%%%%%%%%%%%%%%%%%%%%%%%%%%%%%%%%%%%%%%%%%%%%%%%%%%%%%%%%%%%%%
\subsection{Long-time return memory}
\label{sec:return-memory}
%%%%%%%%%%%%%%%%%%%%%%%%%%%%%%%%%%%%%%%%%%%%%%%%%%%%%%%%%%%%%%%%%%%

For any unit vector
\[
\Psi_0
\in
\mathfrak h_{d-1}\otimes\mathcal K,
\]
define the return probability by
\begin{equation}
R_t(\Psi_0)
:=
\left|
\left\langle
\Psi_0,
\widetilde W_{\mathbf s}^{\,t}\Psi_0
\right\rangle
\right|^2.
\label{eq:harmonic-return-probability}
\end{equation}
Let
\(
\mathsf Q_{\mathrm{har}}
\)
denote the orthogonal projection onto
\(\mathcal H_{\mathrm{har}}\), and define
\begin{equation}
\eta(\Psi_0)
:=
\|\mathsf Q_{\mathrm{har}}\Psi_0\|^2.
\label{eq:harmonic-overlap}
\end{equation}
Thus
\[
0\le\eta(\Psi_0)\le1,
\]
with
\(
\eta(\Psi_0)=1
\)
when
\(\Psi_0\in\mathcal H_{\mathrm{har}}\).

The harmonic subspace gives a quantitative lower bound for the
time-averaged return probability of an arbitrary initial state.

\begin{theorem}
\label{thm:general-harmonic-return-memory}
For every unit
vector
\(
\Psi_0
\in
\mathfrak h_{d-1}\otimes\mathcal K,
\)
we have
\begin{equation}
\lim_{T\to\infty}
\frac1T
\sum_{t=0}^{T-1}
R_t(\Psi_0)
\ge
\begin{cases}
\eta(\Psi_0)^2,
& d=2m,\\[2mm]
\dfrac{\eta(\Psi_0)^2}{8},
& d=2m+1.
\end{cases}
\label{eq:general-harmonic-return-bound}
\end{equation}
\end{theorem}

\begin{proof}
Since \(\mathcal H_{\mathrm{har}}\) is invariant under
\(\widetilde W_{\mathbf s}\) and the restricted dynamics is unitary,
\(\mathcal H_{\mathrm{har}}\) is a reducing subspace.  Hence
\(\mathsf Q_{\mathrm{har}}\) commutes with
\(\widetilde W_{\mathbf s}\).

Let
\[
\widetilde W_{\mathbf s}
=
\sum_{\zeta}
\zeta\,\mathsf F_\zeta
\]
be the spectral resolution of the full walk, where the sum runs over
the distinct eigenvalues and \(\mathsf F_\zeta\) denotes the
corresponding orthogonal spectral projection.  Set
\[
p_\zeta
:=
\|\mathsf F_\zeta\Psi_0\|^2.
\]
Then
\[
\left\langle
\Psi_0,
\widetilde W_{\mathbf s}^{\,t}\Psi_0
\right\rangle
=
\sum_{\zeta}
p_\zeta\zeta^t.
\]
By Ces\`aro averaging,
\begin{equation}
\lim_{T\to\infty}
\frac1T
\sum_{t=0}^{T-1}
R_t(\Psi_0)
=
\sum_{\zeta}p_\zeta^2.
\label{eq:general-return-spectral-formula}
\end{equation}

Suppose first that \(d=2m\).  By
Theorem~\ref{thm:even-harmonic-subspace},
\(\mathcal H_{\mathrm{har}}^{\mathrm{even}}\) is contained in the
eigenspace of \(\widetilde W_{\mathbf s}\) corresponding to the
eigenvalue \(1\).  Hence
\[
p_1
\ge
\|\mathsf Q_{\mathrm{har}}\Psi_0\|^2
=
\eta(\Psi_0).
\]
Therefore
\[
\sum_{\zeta}p_\zeta^2
\ge
p_1^2
\ge
\eta(\Psi_0)^2.
\]

Now let \(d=2m+1\).  Define
\begin{equation}
q_a(\Psi_0)
:=
\left\|
\mathsf E_a
\mathsf Q_{\mathrm{har}}\Psi_0
\right\|^2,
\qquad
1\le a\le8.
\label{eq:general-harmonic-spectral-weights}
\end{equation}
then
\[
\sum_{a=1}^{8}
q_a(\Psi_0)
=
\|\mathsf Q_{\mathrm{har}}\Psi_0\|^2
=
\eta(\Psi_0).
\]
Moreover, since
\(\mathsf E_a\mathsf Q_{\mathrm{har}}\) lies in the full spectral
subspace corresponding to \(\lambda_a\),
\[
p_{\lambda_a}
\ge
q_a(\Psi_0).
\]
It follows from \eqref{eq:general-return-spectral-formula} that
\[
\lim_{T\to\infty}
\frac1T
\sum_{t=0}^{T-1}
R_t(\Psi_0)
\ge
\sum_{a=1}^{8}
q_a(\Psi_0)^2.
\]
By the Cauchy--Schwarz inequality,
\[
\eta(\Psi_0)^2
=
\left(
\sum_{a=1}^{8}q_a(\Psi_0)
\right)^2
\le
8
\sum_{a=1}^{8}
q_a(\Psi_0)^2.
\]
Hence
\[
\lim_{T\to\infty}
\frac1T
\sum_{t=0}^{T-1}
R_t(\Psi_0)
\ge
\frac{\eta(\Psi_0)^2}{8}.
\]
\end{proof}

For initial states entirely contained in the harmonic subspace,
Theorem~\ref{thm:general-harmonic-return-memory} reduces to the
following dimension-independent bounds.

\begin{corollary}
\label{cor:harmonic-return-memory}
If \(d=2m+1\), then for every unit vector
\(
\Phi_0
\in
\mathcal H_{\mathrm{har}}^{\mathrm{odd}},
\)
we have
\begin{equation}
\lim_{T\to\infty}
\frac1T
\sum_{t=0}^{T-1}
R_t(\Phi_0)
=
\sum_{a=1}^{8}q_a^2
\ge
\frac18,
\label{eq:harmonic-mean-return}
\end{equation}
where
\(
q_a
:=
\|\mathsf E_a\Phi_0\|^2,
1\le a\le8.
\)

If \(d=2m\), for every unit vector
\(
\Phi_0
\in
\mathcal H_{\mathrm{har}}^{\mathrm{even}},
\)
we have
\begin{equation}
R_t(\Phi_0)
=
1,
\qquad
t\ge0.
\label{eq:even-return-probability}
\end{equation}
\end{corollary}

%%%%%%%%%%%%%%%%%%%%%%%%%%%%%%%%%%%%%%%%%%%%%%%%%%%%%%%%%%%%%%%%%%%
\subsection{Obstruction to uniform mixing}
\label{sec:absence-uniform-mixing}
%%%%%%%%%%%%%%%%%%%%%%%%%%%%%%%%%%%%%%%%%%%%%%%%%%%%%%%%%%%%%%%%%%%

Let \(\kappa_d\) denote the uniform probability measure on the
\(d\)-dimensional hypercube:
\[
\kappa_d(\sigma)
=
2^{-d},
\qquad
\sigma\in\Gamma_{d-1}.
\]
For probability measures \(\mu\) and \(\nu\) on
\(\Gamma_{d-1}\), write
\[
\|\mu-\nu\|_{\mathrm{TV}}
:=
\frac12
\sum_{\sigma\in\Gamma_{d-1}}
|\mu(\sigma)-\nu(\sigma)|.
\]

Define
\[
\Lambda_d^{\mathrm{har}}
:=
\begin{cases}
\Omega_{m-1}\cup\Omega_m\cup\Omega_{m+1},
& d=2m,\\[1mm]
\Omega_{m-1}\cup\Omega_m\cup
\Omega_{m+1}\cup\Omega_{m+2},
& d=2m+1.
\end{cases}
\]

\begin{theorem}
\label{thm:harmonic-nonmixing}
For every unit vector
\(
\Phi_0\in\mathcal H_{\mathrm{har}},
\)
we have
\[
\left\|
P_t(\,\cdot\mid\Phi_0)-\kappa_d
\right\|_{\mathrm{TV}}
\ge
1-\vartheta_d,
\qquad
t\ge0,
\]
and
\[
\left\|
P_T(\,\cdot\mid\Phi_0)-\kappa_d
\right\|_{\mathrm{TV}}
\ge
1-\vartheta_d,
\qquad
T\ge1,
\]
where
\[
\vartheta_d
:=
\frac{|\Lambda_d^{\mathrm{har}}|}{2^d}
=
\begin{cases}
\displaystyle
\frac{3m+1}{m+1}
\frac{\binom{2m}{m}}{4^m},
& d=2m,\\[4mm]
\displaystyle
\frac{2(2m+1)}{m+2}
\frac{\binom{2m}{m}}{4^m},
& d=2m+1.
\end{cases}
\]
In particular,
\[
\inf_{\substack{
\Phi_0\in\mathcal H_{\mathrm{har}}\\
\|\Phi_0\|=1
}}
\;
\inf_{t\ge0}
\left\|
P_t(\,\cdot\mid\Phi_0)-\kappa_d
\right\|_{\mathrm{TV}}
\longrightarrow
1
\qquad
(d\to\infty).
\]
\end{theorem}

\begin{proof}
By the layer confinement
\eqref{eq:harmonic-layer-confinement},
both \(P_t(\,\cdot\mid\Phi_0)\) and
\(P_T(\,\cdot\mid\Phi_0)\) are supported on
\(\Lambda_d^{\mathrm{har}}\).  Hence
\[
\left\|
P_t(\,\cdot\mid\Phi_0)-\kappa_d
\right\|_{\mathrm{TV}}
\ge
1-\kappa_d(\Lambda_d^{\mathrm{har}})
=
1-\vartheta_d,
\]
and the same estimate holds for \(P_T\).  Finally,
\[
\frac{\binom{2m}{m}}{4^m}
\sim
\frac1{\sqrt{\pi m}},
\]
so \(\vartheta_d\to0\).
\end{proof}

\begin{remark}
Theorem~\ref{thm:harmonic-nonmixing} contrasts with the uniform
distribution arising for the QBN hypercube walk considered in
\cite{WangHypercube2022}.  Here the harmonic dynamics remains confined
to central Hamming layers whose relative volume tends to zero.
Consequently, both the instantaneous and the time-averaged position
distributions remain asymptotically maximally separated from the
uniform measure.
\end{remark}
%%%%%%%%%%%%%%%%%%%%%%%%%%%%%%%%%%%%%%%%%%%%%%%%%%%%%%%%%%%%%%%%%%%
\subsection{Hamming-metric extension}
\label{sec:hamming-metric-extension}
%%%%%%%%%%%%%%%%%%%%%%%%%%%%%%%%%%%%%%%%%%%%%%%%%%%%%%%%%%%%%%%%%%%

The preceding result concerns the relative volume of the position
support.  We now show that small relative volume does not imply
concentration near a single vertex.

For \(\sigma,\tau\in\Gamma_{d-1}\), let
\[
d_H(\sigma,\tau)
:=
|\sigma\triangle\tau|
\]
be the Hamming distance, and write
\[
\operatorname{diam}_H(\operatorname{supp}P)
:=
\max_{\sigma,\tau\in\operatorname{supp}P}
d_H(\sigma,\tau).
\]

We first record a complement relation for the central harmonic
spaces.  Recall that the QBN coordinate shifts satisfy
\(
\Xi_k Z_\sigma
=
Z_{\sigma\triangle\{k\}},
\)
and hence
\begin{equation}
\Xi_0\Xi_1\cdots\Xi_{d-1}Z_\sigma
=
Z_{\sigma^c}.
\label{eq:qbn-total-coordinate-flip}
\end{equation}

\begin{lemma}
\label{lem:harmonic-complement}
If \(d=2m\), then
\begin{equation}
\Xi_0\Xi_1\cdots\Xi_{2m-1}f
=
(-1)^m f,
\qquad
f\in\mathscr H_m^-.
\label{eq:even-qbn-complement}
\end{equation}

If \(d=2m+1\), then
\begin{equation}
\Xi_0\Xi_1\cdots\Xi_{2m}f
=
(-1)^m A^+f,
\qquad
f\in\mathscr H_m^-.
\label{eq:odd-qbn-complement}
\end{equation}
Consequently,
\[
\operatorname{supp}(A^+f)
=
\{
\sigma^c:
\sigma\in\operatorname{supp}f
\}.
\]
\end{lemma}

\begin{proof}
Let first \(d=2m\).  By
\eqref{eq:qbn-total-coordinate-flip}, the operator
\(
\Xi_0\Xi_1\cdots\Xi_{2m-1}
\)
preserves
the central harmonic space \(\mathscr H_m^-\) and commutes with the
natural \(S_d\)-action.  Since
\(
\mathscr H_m^-
\cong
S^{(m,m)}
\)
is irreducible, Schur's lemma implies that its restriction to
\(\mathscr H_m^-\) is scalar.

Choose a partition
\[
\mathbb N_{2m-1}
=
\{a_1,b_1\}
\sqcup\cdots\sqcup
\{a_m,b_m\}
\]
and set
\[
f_0
=
\prod_{j=1}^m
(\partial_{a_j}^*-\partial_{b_j}^*)
Z_\varnothing.
\]
Then
\[
\Xi_0\Xi_1\cdots\Xi_{2m-1}f_0
=
(-1)^m f_0.
\]
Hence the scalar is \((-1)^m\), proving
\eqref{eq:even-qbn-complement}.

Now let \(d=2m+1\).  Both
\[
\Xi_0\Xi_1\cdots\Xi_{2m}
\quad\text{and}\quad
A^+
\]
are \(S_d\)-equivariant isomorphisms from
\(\mathscr H_m^-\) onto \(\mathscr H_{m+1}^+\).
Since
\(
\mathscr H_m^-
\cong
S^{(m+1,m)}
\)
is irreducible, the two maps differ by a scalar.

Choose mutually disjoint pairs
\[
\{a_1,b_1\},\ldots,\{a_m,b_m\},
\]
and let \(c\) be the remaining coordinate. Then
we have
\[
A^+f_0
=
\partial_c^*f_0.
\]
On the other hand,
\[
\Xi_0\Xi_1\cdots\Xi_{2m}f_0
=
(-1)^m\partial_c^*f_0.
\]
Therefore
\[
\Xi_0\Xi_1\cdots\Xi_{2m}f
=
(-1)^mA^+f
\]
for every \(f\in\mathscr H_m^-\), proving
\eqref{eq:odd-qbn-complement}.

Finally, by \eqref{eq:qbn-total-coordinate-flip}, the left-hand side
of \eqref{eq:odd-qbn-complement} replaces each basis vector
\(Z_\sigma\) by \(Z_{\sigma^c}\).  Hence
\[
\operatorname{supp}(A^+f)
=
\{
\sigma^c:
\sigma\in\operatorname{supp}f
\}.
\]
\end{proof}

\begin{theorem}
\label{thm:maximal-hamming-extension}
If \(d=2m\), then every unit vector
\(
\Phi_0\in\mathcal H_{\mathrm{har}}^{\mathrm{even}}
\)
satisfies
\begin{equation}
P_t(\sigma^c\mid\Phi_0)
=
P_t(\sigma\mid\Phi_0),
\qquad
\sigma\in\Gamma_{d-1},\quad t\ge0.
\label{eq:even-antipodal-probability}
\end{equation}
Consequently,
\begin{equation}
\operatorname{diam}_H
\left(
\operatorname{supp}
P_t(\,\cdot\mid\Phi_0)
\right)
=
d,
\qquad
t\ge0,
\label{eq:even-maximal-diameter}
\end{equation}
and, for every \(\eta\in\Gamma_{d-1}\),
\begin{equation}
\sum_{\sigma\in\Gamma_{d-1}}
d_H(\eta,\sigma)
P_t(\sigma\mid\Phi_0)
=
\frac d2.
\label{eq:even-mean-hamming-distance}
\end{equation}

If \(d=2m+1\), let
\[
0\ne f\in\mathscr H_m^-,
\qquad
\mathcal M_m^{(\varepsilon)}x
=
\lambda x,
\qquad
x\ne0.
\]
Then
\begin{equation}
\operatorname{diam}_H
\left(
\operatorname{supp}
P_t\!\left(
\,\cdot\mid
\Psi_f(x)
\right)
\right)
=
d,
\qquad
t\ge0.
\label{eq:odd-maximal-diameter}
\end{equation}
\end{theorem}

\begin{proof}
Let first \(d=2m\).  By
Lemma~\ref{lem:harmonic-complement},
\[
\Xi_0\Xi_1\cdots\Xi_{2m-1}f
=
(-1)^m f,
\qquad
f\in\mathscr H_m^-.
\]
Since the product of all coordinate shifts sends each hypercube vertex
\(\sigma\) to its antipode \(\sigma^c\), the three-layer construction
\eqref{eq:three-layer-state} gives
\eqref{eq:even-antipodal-probability}.
If
\(
\sigma\in
\operatorname{supp}
P_t(\,\cdot\mid\Phi_0),
\)
then \(\sigma^c\) also belongs to the support.  Since
\(
d_H(\sigma,\sigma^c)=d,
\)
we obtain
\eqref{eq:even-maximal-diameter}.
Moreover,
\(
d_H(\eta,\sigma)
+
d_H(\eta,\sigma^c)
=
d.
\)
Combining this identity with the antipodal symmetry and normalization
gives
\eqref{eq:even-mean-hamming-distance}.

Now let \(d=2m+1\), and set
\(
g:=A^+f.
\)
By Lemma~\ref{lem:harmonic-complement},
\[
\operatorname{supp}g
=
\{
\sigma^c:
\sigma\in\operatorname{supp}f
\}.
\]
Let
\(
\mathcal M_m^{(\varepsilon)}x
=
\lambda x.
\)
The reduced equations \eqref{eq:odd-reduced-map} show that neither of
the two harmonic partners \(f\) and \(g\) can disappear from
\(\Psi_f(x)\).  Indeed, if either partner were absent, the
corresponding coefficient equations would imply
\(
\lambda^2=-1,
\)
contrary to \eqref{eq:no-fourth-root-spectrum}.
Hence the position support of \(\Psi_f(x)\) contains a complementary
pair
\(
\sigma,\sigma^c.
\)
Since
\(
d_H(\sigma,\sigma^c)=d,
\)
the position support has Hamming diameter \(d\) at \(t=0\).
Finally,
\[
\widetilde W_{\mathbf s}^{\,t}\Psi_f(x)
=
\lambda^t\Psi_f(x),
\]
so the position support is unchanged by the evolution.  Hence
\eqref{eq:odd-maximal-diameter} holds for every \(t\ge0\).
\end{proof}

%%%%%%%%%%%%%%%%%%%%%%%%%%%%%%%%%%%%%%%%%%%%%%%%%%%%%%%%%%%%%%%%%%%
\subsection{Sparse transversal harmonic states}
\label{sec:transversal-states}
%%%%%%%%%%%%%%%%%%%%%%%%%%%%%%%%%%%%%%%%%%%%%%%%%%%%%%%%%%%%%%%%%%%

Let
\[
d=2m+\delta,
\qquad
\delta\in\{0,1\},
\]
and choose mutually disjoint pairs
\[
\{a_1,b_1\},\ldots,\{a_m,b_m\}.
\]

For
\[
f
=
\sum_{\sigma\in\Omega_m}
f_\sigma Z_\sigma
\in V_m,
\qquad
f_\sigma
=
\langle f,Z_\sigma\rangle,
\]
we write
\[
\operatorname{supp}f
:=
\{\sigma\in\Omega_m:f_\sigma\ne0\}.
\]

Define
\[
f^{\mathrm{tr}}
:=
\prod_{j=1}^{m}
(\partial_{a_j}^*-\partial_{b_j}^*)
Z_\varnothing.
\]
Then
\[
A^-f^{\mathrm{tr}}=0,
\qquad
|\operatorname{supp}f^{\mathrm{tr}}|
=
2^m.
\]
In odd dimension, if \(c\) denotes the unpaired coordinate, then
\[
A^+f^{\mathrm{tr}}
=
\partial_c^*f^{\mathrm{tr}}.
\]

The QBN harmonic kernel also has a Johnson-graph interpretation.  On
\(V_m\),
\[
A^+A^-
=
mI+A_{J(d,m)},
\]
where \(A_{J(d,m)}\) is the adjacency operator of the Johnson graph
\(J(d,m)\).  Hence
\[
\mathscr H_m^-
=
\ker\bigl(A_{J(d,m)}+mI\bigr).
\]
The classical minimum-support bound gives
\[
|\operatorname{supp}f|
\ge
2^m,
\qquad
0\ne f\in\mathscr H_m^-,
\]
and the transversal vector \(f^{\mathrm{tr}}\) attains equality
\cite{Hwang1986,VorobevMogilnykhValyuzhenich2018}.

\begin{corollary}
\label{cor:transversal-sparsity}
Assume \eqref{eq:even-supporting-spin} when \(d=2m\), and
\eqref{eq:odd-supporting-spin} when \(d=2m+1\).

If \(d=2m\), then
\[
\frac{
\left|
\operatorname{supp}
P_t\!\left(
\,\cdot\mid
\Phi_{f^{\mathrm{tr}}}^{(\varepsilon)}
\right)
\right|
}{
2^{2m}
}
=
(m+1)2^{-m},
\qquad
t\ge0.
\]

If \(d=2m+1\) and
\[
\mathcal M_m^{(\varepsilon)}x
=
\lambda x,
\qquad
x\ne0,
\]
then
\[
\frac{
\left|
\operatorname{supp}
P_t\!\left(
\,\cdot\mid
\Psi_{f^{\mathrm{tr}}}(x)
\right)
\right|
}{
2^{2m+1}
}
\le
(m+1)2^{-m},
\qquad
t\ge0.
\]
In both cases the position support has Hamming diameter \(d\).
\end{corollary}

Indeed, in even dimension the three-layer lift contains
\[
2^m+2m2^{m-1}
=
(m+1)2^m
\]
position vertices.  In odd dimension the four supporting layers
contain at most
\[
m2^{m-1},\qquad
2^m+m2^{m-1},\qquad
2^m+m2^{m-1},\qquad
m2^{m-1}
\]
possible nonzero vertices, whose total is
\[
2^{m+1}(m+1).
\]
The Hamming-diameter statement follows from
Theorem~\ref{thm:maximal-hamming-extension}.

Thus the transversal family realizes the extreme combination
\[
\frac{|\operatorname{supp}P_t|}{2^d}
\longrightarrow0
\]
at an exponential rate, while
\[
\operatorname{diam}_H(\operatorname{supp}P_t)=d.
\]
In particular, strong participation sparsity is compatible with
macroscopic extension in Hamming distance.
\section{Conclusion}
\label{sec:conclusion}

In this paper, we have developed a QBN harmonic reduction for
spin-coupled Grover walks on hypercubes under the assumption that the
spin variables take a common value on the central supporting Hamming
layers, and obtained a parity-dependent description of the resulting
invariant dynamics.
For \(d=2m\), the central harmonic kernel generates a three-layer
fixed-point subspace of dimension \(C_m\), and for \(m\ge2\) this
subspace exhausts all fixed points supported on the three central
Hamming layers.  For \(d=2m+1\), the two central harmonic spaces are
connected by the collective creation operator, and their canonical
QBN lifts generate a four-layer invariant subspace of the form
\[
\mathcal H_{\mathrm{har}}^{\mathrm{odd}}
\cong
\mathscr H_m^-\otimes\mathbb C^8.
\]
Consequently, the nontrivial part of the dynamics on this invariant
subspace is given by an explicit eight-dimensional reduced matrix.

In odd dimension, the central harmonic kernel is identified with the
Specht module \(S^{(m+1,m)}\), and the resulting four-layer harmonic
subspace is precisely the full corresponding isotypic component.
Moreover, this representation occurs with multiplicity \(8\), exactly
matching the dimension of the reduced dynamical factor.

As consequences of this reduction, we also obtain explicit long-time
probability behavior and show that harmonic states can remain strongly
nonuniform and occupy an asymptotically vanishing fraction of the
hypercube while still having maximal Hamming diameter.

These results provide a systematic QBN description of the central
invariant dynamics of spin-coupled Grover walks in arbitrary dimension
and suggest two natural directions for further study: extending the
harmonic reduction to more general spin configurations, and developing
a complete spectral classification of the full walk, including the
structure and interaction of the remaining invariant spectral
subspaces.

\section*{Author declarations}

\subsection*{Conflict of interest}

The author has no conflicts to disclose.

\subsection*{Data Availability}

Data sharing is not applicable to this article as no new data were created or analyzed in this study.

\end{document}